\documentclass[12pt]{article}
\usepackage[utf8]{inputenc}
\usepackage{amsmath}
\newcommand\numberthis{\addtocounter{equation}{1}\tag{\theequation}}
\usepackage{amssymb}
\usepackage{amsthm}
\usepackage{mathtools}
\usepackage{color}
\usepackage{bigints}
\usepackage{setspace}
\usepackage{enumerate}
\usepackage{comment}
\usepackage{natbib}
\usepackage{tikz}
\usepackage{pgfplots}
\usepackage{subcaption}
\usepackage{graphicx,pgfarrows,pgfnodes}
\usepackage[justification=centering]{caption}
\usepackage{babel}
\usepackage{array}
\usepackage{verbatim}

\allowdisplaybreaks
\usepackage[top=1.45in, bottom=1.45in, left=1.25in, right=1.25in]{geometry}
\title{Posterior-Mean Separable Costs of Information Acquisition}

\usepackage{hyperref}
\newtheorem{axiom}{\sc Axiom}
\newtheorem{defn}{\sc Definition}
\newtheorem{theorem}{\sc Theorem}
\newtheorem{lemma}{\sc Lemma}
\newtheorem{proposition}{\sc Proposition}

\title{{\bf Posterior-Mean Separable Costs of Information Acquisition}%
\footnote{We acknowledge support from the Israel Science Foundation 
(grant 2372/21). We would like to thank Tommaso Denti, Yoram Halevy, 
Sergiu Hart and Doron Ravid, as well as seminar participants at Hebrew 
University, the University of Bonn, Columbia University, Northwestern 
University, Arizona State University, the Southern Economic Conference 2023, 
and EC 2024 for helpful conversations and suggestions.}}

\author{
\begin{minipage}{0.3\textwidth}\centering  
Jeffrey Mensch\footnote{\texttt{jeffrey.mensch@mail.huji.ac.il, https://sites.google.com/site/jeffreyimensch/.}} \\ \centering \it \small Hebrew University of Jerusalem
\end{minipage}                  
\begin{minipage}{0.3\textwidth}\centering 
Komal Malik\footnote{\texttt{komal.malik@mail.huji.ac.il, https://sites.google.com/view/komal-malik/home}}  \\ \centering \it \small Hebrew University of Jerusalem
\end{minipage} 
}

\date{\vspace{0.8cm} \today}
\begin{document}

\maketitle
\begin{abstract}

\noindent 

We analyze a problem of revealed preference given state-dependent
stochastic choice data in which the payoff to a decision maker (DM) only 
depends on their beliefs about posterior means. Often, the DM must also
learn about or pay attention to the state; in applied work on this 
subject, a convenient assumption is that the costs of such learning are 
linearly dependent in the distribution over posterior means. We provide 
testable conditions to identify whether this assumption holds. This allows 
for the use of information design techniques to solve the 
DM's problem.

\vspace{.2cm}


\end{abstract}

\pagebreak{}

\section{Introduction}

Information plays an important role in the making of decisions, as the 
optimal choice under uncertainty will depend on how beliefs are updated
from the signal. This paper focuses specifically on decision problems in 
which the payoff to the decision maker (DM) depends only on the 
\emph{posterior mean} of the underlying state. These problems are of 
special interest, since they include environments where the DM is risk-%
neutral over outcomes (such as in many standard models of financial 
markets), or where there is uncertainty over the value of an item for 
sale in an auction. 

Our question, then, is to develop a revealed-preference approach under 
which one can represent the choices of the DM as arising from a cost of 
information acquisition that only depends separably on posterior means.
This characterization allows for an understanding of the behavioral 
implications of this class of information cost functions, by seeing what
restrictions this assumption places on decision patterns. It also allows
this assumption to be tested empirically in datasets by testing for 
consistency with the assumption.

We build on previous work by \cite{caplin2015revealed} and 
\cite{denti2022posterior}, who also use a revealed preference 
approach to represent information costs. In particular, \cite{caplin2015revealed} 
show that under their two main assumptions, \emph{no improving action 
switches} (NIAS) and \emph{no improving attention cycles} (NIAC), one 
can represent the preferences of the DM as coming from a cost of 
information acquisition $C(\mu)$ that is \emph{additively separable} 
from the relevant decision problem, where $\mu$ is the distribution of 
posterior beliefs. \cite{denti2022posterior} extended these results, 
showing that by replacing NIAC with a related assumption, \emph{no 
improving posterior cycles} (NIPC), it is possible to represent $C(\mu)$ 
as being \emph{posterior-separable}: 
\begin{equation}
C(\mu)=H(\pi)-\int H(p)d\mu(p) \label{postsepcosts}
\end{equation}
where $\pi$ is the prior belief, $p$ are the posterior beliefs, and $H$ is 
a concave function. In the present work, we replace NIPC with a stronger 
axiom, \emph{no improving posterior-mean cycles} (NIPMC), under which the 
cost of information acquisition can be represented as \emph{posterior-mean 
separable}: given prior mean $z_0\in \mathbb{R}$ and distribution of 
posterior means $F$,
\begin{equation}
C(F)=c(z_0)-\int c(z)dF(z) \label{pmsepcosts}
\end{equation}
where $c:\mathbb{R}\rightarrow \mathbb{R}$ is a concave function. As shown 
by \cite{rothschild1970increasing} and \cite{gentzkow2016rothschild},
$F$ must be a \emph{mean-preserving contraction} of the prior distribution 
$F_0$.

This class of cost functions was introduced by \cite{ravid2022learning}, 
and has been subsequently used in applications by 
\cite{mensch2024monopoly}, \cite{kreutzkamp2022endogenous}, 
\cite{whitmeyer2022costly}, and \cite{thereze2022adverse}. While 
more restrictive than the class of posterior-separable cost 
functions that have become prominent in the literature 
\citep{caplin2013behavioral,caplin2022rationally}, this class 
covers such natural cost functions as those proportional to the
reduction of posterior variance. Let $\mathbf{z}_{F_0}$ and 
$\mathbf{z}_F$ be random variables distributed according to $F_0$ and 
$F$, respectively, where $F$ is the distribution of posterior beliefs 
about $z_{F_0}$ from some signal. Then by the law of total variance,
\[
Var(\mathbf{z}_{F_0})=E[Var(\mathbf{z}_{F_0}\vert \mathbf{z}_F=z)]+Var(E[\mathbf{z}_{F_0}\vert \mathbf{z}_F=z])
\]
\[
=E[Var(\mathbf{z}_{F_0}\vert \mathbf{z}_F=z)]+Var(\mathbf{z}_F)
\]
since $E[\mathbf{z}_{F_0}\vert \mathbf{z}_F=z]=z$. This implies that the 
expected reduction in variance is 
\[
Var(\mathbf{z}_F)=\int_{\underline{z}}^{\bar{z}}(z-z_0)^2 dF(z)
\]
Thus, costs of information acquisition proportional to the reduction of 
posterior variance yield a representation within our class of the form 
$c(z)=-\kappa(z-z_0)^2$, for $\kappa>0$.

This class is also particularly useful in practical applications, as it 
allows the use of techniques from the literature on information design in 
determining the DM's optimal information choices. An appealing feature
of the use of posterior-separable cost functions is that the DM can view
the entire decision problem, both of acquiring information and then 
subsequently acting on it, as a function of the posterior belief. This 
allows for the use of concavification techniques 
\citep{Aumann1995,Kamenica2011} to solve for the DM's optimal information 
acquisition \citep{caplin2013behavioral,matvejka2015rational,caplin2022rationally}. 
When both the decision problem and the costs of information acquisition 
only depend on posterior \emph{means}, then the optimal choice of the DM 
can be determined by finding a convex function $P$ that minimizes the 
expectation with respect to the prior subject to $P$ always lying above 
the objective, as a function of the posterior mean 
\citep{dworczak2019simple}. One can then use additional 
techniques, such as majorization/bi-pooling 
\citep{kleiner2021extreme,arieli2020optimal} to obtain a more
precise characterization of the solution. Indeed, in 
\cite{mensch2024monopoly}, these tools were used to show that a
price-discriminating monopolist, \'{a} la \cite{mussa1978monopoly}, 
distorts quality downward even at the top, and offers simple menus 
when the marginal cost of information is too high. Thus the use of 
posterior-mean measurable costs of information allows for a tractable 
approach to solving models with information acquisition for decision problems whose payoffs only depend on posterior means.

While our characterization approach is similar to that of 
\cite{denti2022posterior}, care must be taken to modify the steps in a 
way that is appropriate for our alternative axiom. In particular, Denti's
NIPC axiom relies on the observation that the optimal information choice 
is given by the concavification of the DM's objective. Thus, the 
intuition is that, if one reallocates posteriors between decision 
problems, the resulting posteriors may fall below the envelope (namely, a 
hyperplane) characterizing the optimal information choice for each 
decision problem, and so be suboptimal. For mean-measurable information 
design, though, this envelope is not given by a hyperplane, but by a 
convex function which is only affine on intervals where the 
mean-preserving contraction constraint on $G_\mu$ does not bind. As a 
result, a distribution of posteriors $\mu$ is only optimal if there are 
no mean-preserving spreads available on strictly convex regions of the 
DM's objective. Our NIPMC axiom must therefore take into account the 
locations where this constraint binds when comparing reallocations of 
posterior means.

This paper contributes to the literature on representation of costly 
information acquisition. Aside from the aforementioned papers by 
\cite{caplin2015revealed}, \cite{caplin2022rationally}, and 
\cite{denti2022posterior}, other papers that provide characterizations for a 
large class of cost functions include \cite{de2017rationally}, 
\cite{mensch2018cardinal}, \cite{ellis2018foundations}, \cite{chambers2020costly}, 
\cite{bloedel2020cost}, \cite{denti2022experimental}, 
\cite{hebert2023rational}, and \cite{lin2022stochastic}. Other papers provide foundations 
for more specific costs of information acquisition, such as entropy reduction 
\citep{de2014axiomatic,mensch2021rational,caplin2022rationally}, log-%
likelihood ratio costs \citep{pomatto2023cost}. Our work builds on these 
by providing foundations for posterior-mean separable information costs,
a convenient class of functions for many economic problems.

We also contribute to the literature on information design based on 
posterior means. As noted above, \cite{dworczak2019simple} characterize 
the solution for the optimal information design problem for payoffs that 
only depend on posterior means; their results were extended under slightly 
weaker assumptions by \cite{dizdar2020simple}. Other work on information 
design for posterior means includes \cite{gentzkow2016rothschild}, 
\cite{kolotilin2017persuasion}, \cite{dworczak2019persuasion}, 
\cite{kolotilin2019persuasion}, and \cite{zapechelnyuk2020optimal}.

The rest of the paper proceeds as follows. Section 2 presents the model 
and preliminary constructions that will be useful for the results and 
discussion in the remainder of the paper. Section 3 presents our axioms 
on the dataset and our main theorem. Section 4 briefly discusses the 
example of costs proportional to posterior variance. Section 5 presents 
further discussion and extensions of our main results. Section 6 
concludes.

\section{Preliminaries}

\subsection{Model}

Our approach very closely follows that of \cite{denti2022posterior}, though some
of the notation has been modified, and on occasion is closer to that of 
\cite{caplin2015revealed} and \cite{chambers2020costly}.

We examine the choice of a decision maker (DM) from the perspective of an 
analyst who observes state-dependent stochastic choice data. Let 
$Z\subset\mathbb{R}$ be a finite%
\footnote{We extend the model to infinite states and acts in Section 5.}
set of states of nature; without loss
of generality, we normalize the lowest state $\underline{z}$ to $0$ and 
the highest state $\bar{z}$ to $1$.  We let $X\subset\mathbb{R}$ be the 
(finite) grand set of acts, from which finite menus $A\subset X$ for 
the individual decision problems are drawn. We then define the finite set 
$\mathcal{A}$ as the finite set of decision problems (menus) $A$ that the DM 
can actually face. For each such $A$, the analyst observes the state-dependent 
stochastic choice (SDSC) function of the DM, $\sigma_A:Z\rightarrow\Delta(A)$. 
Let $S_{A}$ be the set of all such stochastic choice functions.

\begin{defn}
A \emph{state-dependent stochastic choice dataset} is a finite collection of 
decision problems $\mathcal{A}$, with, for each $A\in\mathcal{A}$, associated 
stochastic choice functions $\sigma_A\in S_A$.
\end{defn}

We assume that there is a prior CDF $F_0$ over $Z$, with associated prior mean 
$z_0$ and probability of each $z\in Z$ denoted by $f_0(z)$, and such that 
$\{0,1\}\subset\mbox{supp}(F_0)$. The von Neumann-Morgenstern utility function 
$u:A\times [0,1]\rightarrow\mathbb{R}$ only depends on the state through the 
posterior mean $z\in [0,1]$. We assume that both the utility 
function and prior distribution are common to all decision problems (with the 
latter drawn i.i.d.) and are known to the analyst.

We focus our analysis on the hypothesis that the DM is maximizing expected
utility, subject to acquiring information with respect to the prior whose
cost is posterior-mean separable. To that end, we provide definitions and 
notation as follows. For each $A\in \mathcal{A}$, the DM chooses an 
\emph{information structure}, which is defined as a distribution of posterior 
beliefs $\mu\in\Delta(\Delta(Z))$. As we are interested in particular in the 
distribution of posterior means, we let $\mathcal{F}$ be the set of CDFs 
over $[0,1]$; this is easily calculated from $\mu$. In order to be a valid 
information structure (i.e. one that satisfies Bayes' rule), 
\cite{gentzkow2016rothschild} showed that the CDF $F$ must be a \emph{mean-%
preserving contraction} of the prior. That is,
\begin{equation}
\label{MPC}
\int_0^z [F_0(s)-F(s)]ds\geq0,\forall z
\end{equation}
with equality at $z=1$. Define
\begin{align*}
I_{F_0,F}: [0,1] & \rightarrow \mathbb{R}\\
z & \rightarrow \int_0^z [F_0(s)-F(s)]ds
\end{align*}
We let $\mathcal{I}_{F_0}$ be the set of feasible posterior-mean distributions
given prior $F_0$:
\begin{equation}
    \label{feasiblecdfs}
    \mathcal{I}_{F_0}=\{F\in \mathcal{F}: I_{F_0,F}(z)\geq0,\forall z\in [0,1], \mbox{ and } I_{F_0,F}(1)=0\}
\end{equation}
We endow $\mathcal{I}_{F_0}$ with the weak$^*$ topology, and assign 
it the mean-preserving spread order $\succeq$. We use 
$\delta_{z}\in\mathcal{F}$ to refer to the distribution of posterior 
means which places probability $1$ at $z$. Thus 
$\delta_{z_0}\in \mathcal{I}_{F_0}$ is the degenerate distribution of 
posteriors which provides no information.

Given menu $A$ and posterior mean $z$, we denote by $\phi_A(z)$ the indirect
utility from the optimal choice of $a\in A$ given $z$:
$$\phi_A(z)=\max_{a\in A} u(a,z)$$
Since $u$ only depends on $z$ through the posterior mean, and is an expected
utility function, it follows that $u(a,z)=zu(a,1)+(1-z)u(a,0)$. Therefore,
$u$ is continuous and affine in $z$. It follows by Berge's maximum theorem that 
$\phi_A$ is also continuous in $z$; as the maximum over affine functions, 
$\phi_A$ is convex as well.

Define the function $C:\mathcal{I}_{F_0}\rightarrow \mathbb{R}_+$ as the DM's
\emph{cost of information acquisition}. The cost is assumed to be the same for
all decision problems. 
\begin{defn}
    The cost function $C$ is \emph{canonical} if, for all 
    $F,\hat{F}\in\mathcal{I}_{F_0}$:
    \begin{enumerate}[(i)]
        \item \textbf{Monotonicity:} If $F\succeq \hat{F}$, then $C(F)\geq C(\hat{F})$
        \item \textbf{Convexity:} $C(\alpha F+(1-\alpha)\hat{F})\leq \alpha C(F)+(1-\alpha)C(\hat{F})$, $\forall \alpha[0,1]$
        \item \textbf{Null experiment:} $C(\delta_{z_0})=0$
    \end{enumerate}
\end{defn}
As shown by \cite{de2017rationally}, any non-canonical cost of information
acquisition can be replaced without loss by a canonical cost function that
generates the same choices.

Given menu $A$, the DM's information acquisition problem is 
\begin{equation}
    \max_{F\in\mathcal{I}_{F_0}} \int_0^1 \phi_A (z)dF(z)-C(F)
\end{equation}
We assume that $C$ is lower-semicontinuous to ensure that the above problem has 
a well-defined solution. As the decision problem given the choice of information 
does not interact with the costs of information acquisition, one can view the 
decision problem as comprising two steps: first, the DM acquires information, 
and then given his posterior, he optimally chooses $a$.

Of particular interest will be the class of information acquisition costs 
that are \emph{posterior-mean separable}.
\begin{defn}
    A cost function $C:\mathcal{I}_{F_0}\rightarrow(-\infty,\infty)$ is 
    \emph{posterior-mean separable} if there is an upper-semicontinuous 
    function $c:[0,1]\rightarrow \mathbb{R}$ such that 
    \begin{equation}
        C(F)=c(z_0)-\int_0^1 c(z)dF(z)
    \end{equation}
\end{defn}
\cite{ravid2022learning} refer to $c$ as the \emph{derivative} of the cost 
function; our case corresponds to that of ``constant marginal costs" in
their paper (Example 2). They further show (Claim 1) that if $C$ is canonical 
and is posterior-mean separable, then $c$ is concave.

The \emph{decision function} $D_A: [0,1]\rightarrow \Delta(A)$ defines the 
probability of the choice of $a$ given the \emph{posterior mean},
which we assume is measurable with respect to $z$. This 
is related to the stochastic choice $\sigma_A$ via Bayes' rule:
\begin{equation}
\label{Bayes}
\sum_{z\in Z} \sigma_A(a\vert z)f_0(z)=\int_0^1 D_A(a\vert \tilde{z})dF(\tilde{z})
\end{equation}

Thus $\sigma_A$ defines the conditional probability of choosing $a$ given the 
\emph{realized} state, whereas $D_A$ defines the conditional probability given
the \emph{posterior mean}.

Let $\mathcal{D}_A$ be the set of measurable functions from $Z$ to $\Delta(A)$. 
Then $(F,D_A)$ is optimal for the DM if 
\begin{equation}
\label{optimality}
\int_0^1\sum_{a\in A} u(a,z)D_A(a\vert z)dF(z)-C(F)\geq \int_0^1 \phi_A(z)d\hat{F}(z)-C(\hat{F}),\forall \hat{F}\in\mathcal{I}_{F_0}
\end{equation}

\begin{defn}
    The cost function $C$ \emph{rationalizes} a SDSC dataset 
    $\{\sigma_A\}_{A\in\mathcal{A}}$ if for every $A\in \mathcal{A}$, $\sigma_A$
    is optimal for the DM (as defined in (\ref{optimality})) and 
    $(\sigma_A,D_A)$ are related by (\ref{Bayes}) given $F$.
\end{defn}

\subsection{Revealed Posterior Means}

For any given $A\in\mathcal{A}$, the relationship between $\sigma_A$
and $D_A$ is defined in (\ref{Bayes}). However, while $D_A$ pins down
$\sigma_A$ given the distribution of posterior means, the reverse 
determination does not hold. That is, one cannot infer what
$D_A$ is from $\sigma_A$, as there could be multiple ways in which the DM
might have $a$ in the support of his choice given his posterior mean which
yield the same empirical SDSC function $\sigma_A$, via different choices of
distributions of posterior means $F$.

One can circumvent this issue by looking at the information acquisition choice
$F$ that requires the \emph{least} information, using an idea that has become
commonplace in the literature on information design. Informally, this can be 
defined from any distribution of posterior means and decision function by 
merging all posterior means at which $a$ is chosen, weighted by the probability 
that $a$ is chosen according to $D_A$. In the resultant distribution, each $a$ 
is only chosen at a single posterior mean. This allows us to define a 
\emph{revealed} decision function.

To accomplish this, we define the notion of \emph{revealed posterior means}, 
adapting the notion of ``revealed posteriors" of \cite{caplin2015testable}
to our setting. For any $a\in A$, let
\[
\sigma_A(a)\coloneqq \sum_{z\in Z} \sigma_A(a\vert z)f_0(z)
\]
be the unconditional probability that $a\in A$ is chosen according to 
$\sigma_A$. 

\begin{defn}
\label{revealedmean}
    For each $a\in \mbox{supp}(\sigma_A)$, the \emph{revealed posterior mean} $z_{\sigma_A}(a)$ 
    is given by 
    \begin{equation}
    z_{\sigma_A}(a)=\frac{\sum_{z\in Z} z\sigma_A(a\vert z)f_0(z)}{\sum_{z\in Z} \sigma_A(a\vert z)f_0(z)}
    \end{equation}
    For $a\notin \mbox{supp}(\sigma_A)$, adopt the 
    convention that $z_{\sigma_A}(a)=z_0$.
\end{defn}

The definition above allows us to construct a CDF of the revealed posterior 
means, which we indicate by 
\[
F_{\sigma_A}(z)=\sigma_A(\{a:z_{\sigma_A}(a)\leq z\})
\]
This CDF has finite support, since $A$ is finite. Thus we
can write the probability of a particular revealed posterior mean as
\[
f_{\sigma_A}(z)\coloneqq \sum_{a:z_{\sigma_A}(a)=z} \sigma_{A}(a)
\]

\begin{defn}
\label{revealeddecision}
    Given $\sigma_A$, the \emph{revealed decision function} 
    $D_{\sigma_A}\in\mathcal{D}_A$ is defined as follows:
    \begin{enumerate}[(i)]
        \item If $z\in\mbox{supp}(F_{\sigma_A})$ and $z=z_{\sigma_A}(a)$, then
        \[
        D_{\sigma_A}(a\vert z)f_{\sigma_A}(z)=\sigma_A(a)
        \]
        \item If $z\in\mbox{supp}(F_{\sigma_A})$ and $z\neq z_{\sigma_A}(a)$,
        then $D_{\sigma_A}(a\vert z)=0$;
        \item If $z\notin\mbox{supp}(F_{\sigma_A})$, then $D_{\sigma_A}(a\vert z)=\sigma_A(a)$.
    \end{enumerate} 
\end{defn}

In particular, if each act $a$ has a distinct revealed posterior mean, then 
$f_{\sigma_A}(z_{\sigma_A}(a))=\sigma_A(a)$ and $D_{\sigma_A}(a\vert z_{\sigma_A}(a))=1$.

We adapt Lemma 1 of \cite{caplin2015revealed} and \cite{denti2022posterior} in 
the following lemma using our notation.

\begin{lemma}
\label{sufficiency}
    For any $F\in \mathcal{I}_{F_0}$, $\sigma_A\in S_A$, and $D_A\in\mathcal{D}_A$,
    \begin{enumerate}[(i)]
        \item If $F=F_{\sigma_A}$ and $D_A=D_{\sigma_A}$, then (\ref{Bayes})
        holds.
        \item If (\ref{Bayes}) holds, then $F\succeq F_{\sigma_A}$.
    \end{enumerate}
\end{lemma}

As in \cite{caplin2015revealed} and \cite{denti2022posterior}, we use Lemma 
\ref{sufficiency} to show that it is without loss for the DM to consider 
the class of CDFs of revealed posterior means when doing so saves on
information costs. For instance, if $C$ is canonical, then less information
is less costly, and so $C(F_{\sigma_A})\leq C(F)$.

\begin{lemma}
\label{revealedoptimal}
    If $(F,D_A)\in \mathcal{I}_{F_0}\times\mathcal{D}_A$ is optimal for 
    the DM, and $C(F_{\sigma_A})\leq C(F)$, then 
    $(F_{\sigma_A},D_{\sigma_A})$ is optimal as well.
\end{lemma}

This is immediate from the fact that, by Lemma \ref{sufficiency}, the 
induced joint distribution of states and acts is the same in 
$(F_{\sigma_A},D_{\sigma_A})$ as in $(F,D_A)$, while the former involves 
less (costly) information acquisition. Therefore, 
$(F_{\sigma_A},D_{\sigma_A})$ must be optimal as well. 

\section{Axiomatization of Posterior-Mean Separability}

Our characterization of posterior-mean separable costs of information
acquisition will follow from the following two axioms on SDSC data.

\begin{axiom}[No Improving Action Switches]
\label{NIAS}
    For every $A\in\mathcal{A}$ and every $a\in\mbox{supp}(\sigma_A)$,
    and $b\in A$,
    \begin{equation}
        \label{NIASeq}
        u(a,z_{\sigma_A}(a))\geq u(b,z_{\sigma_A}(a))
    \end{equation}
\end{axiom}

This axiom (abbreviated as NIAS), introducted in \cite{caplin2015testable}, 
states that the act chosen at each revealed posterior mean is optimal. This 
can be expressed in terms of SDSC data $\sigma_A$ via the transformation in 
(\ref{Bayes}), though the current form of expression in terms of revealed 
posterior means makes the meaning of the axiom clearer.%
\footnote{\cite{denti2022posterior} presents his version of NIAS in terms
of revealed posteriors as well.}
NIAS states that the choice of $a$ in menu $A$ is optimal conditional on
having posterior mean $z=z_{\sigma_A}(a)$. Thus, $a$ is consistent with 
choosing based on this information. 

The second axiom is a strengthening of ``No Improving Attention Cycles" 
\citep{caplin2015revealed} and ``No Improving Posterior Cycles" 
\citep{denti2022posterior}. Informally, the idea is that the DM cannot
improve his payoff by ``reallocating" individual revealed posterior means 
across decision problems. We define this idea of reallocation below.

\begin{defn}
    For $i\in\{1,...,n\}$, given $\beta_i\in\mathbb{R}_+$, priors 
    $F_0^i\in \mathcal{F}$ such that $\mbox{supp}(F_0^i)\subset Z$,
    and revealed distributions of posterior means $F_i\in\mathcal{I}_{F_0^i}$,
    the sequence $G_i\in\mathcal{I}_{F_0^i}$  is a \emph{reallocation of posterior means}
    if for all $z\in [0,1]$,
    \begin{equation}
        \label{reallocation}
        \sum_{i=1}^n \beta_i F_i(z)= \sum_{i=1}^n \beta_i G_i(z)
    \end{equation}
\end{defn}

A reallocation of posterior means maintains Bayesian consistency with 
respect to each prior $F_0^i$, while also maintaining the weighted average
of the distribution of \emph{posterior means} across all $i$. Within these 
constraints, the way that posterior means are distributed can be shifted 
across decision problems.

\begin{axiom}[No Improving Posterior Mean Cycles]
    \label{NIPMC}
    For $i\in\{1,...,n\}$, consider the sequence of menus
    $A_i\in\mathcal{A}$, coefficients $\beta_i\in \mathbb{R}_+$, priors 
    $F_0^i\in \mathcal{F}$ such that $\mbox{supp}(F_0^i)\subset Z$, and revealed 
    distributions of posterior means $F_i\in \mathcal{I}_{F_0^i}$ such that:
    \begin{enumerate}[(i)]
        \item $\mbox{supp}(F_i)\subset \mbox{supp}(F_{\sigma_{A_i}})$, and
        \item For all $z\in [0,1]$, 
        \begin{equation}
        \label{allineq}
            I_{F_0,F_{\sigma_{A_i}}}(z)=0\implies I_{F^i_0,F_i}(z)=0
        \end{equation}
    \end{enumerate} 
    Then for every reallocation of posterior means $\{G_i\}_{i=1}^n$,
    \begin{equation}
    \label{reallocineq}
        \sum_{i=1}^n \beta_i \int \phi_{A_i}(z)dF_i(z)\geq \sum_{i=1}^n \beta_i \int \phi_{A_i}(z)dG_i(z)
    \end{equation}
\end{axiom}

In words, \emph{No Improving Posterior-Mean Cycles} (NIPMC) says the 
following. Suppose we start with priors $F_0^i$; these need not be 
identical to $F_0$, but rather represent what would happen 
counterfactually if these were the priors. Suppose further that the 
mean-preserving contraction constraint holds for each $F_i$ with respect 
to prior $F_0^i$ at the same points as $F_{\sigma_{A_i}}$ does with 
respect to $F_0$, and the support of $F_i$ is contained within the 
support of $F_{\sigma_{A_i}}$. In such a scenario, reallocating 
individual posterior means across decision problems in a way that 
maintains Bayesian consistency at each prior $F_0^i$ cannot increase 
the DM's payoff.

It is readily seen that NIPMC is a strengthening of \cite{denti2022posterior}'s 
``No Improving Posterior Cycles" (NIPC) axiom. Informally, Denti's NIPC 
axiom considers posteriors rather than posterior means, and drops condition 
(\ref{allineq}); under such conditions, the equivalent of (\ref{reallocineq})
holds for distributions over posteriors. Notice that, if distributions of 
posteriors $\mu_i$ have the same supports as $\mu_{\sigma_{A_i}}$, then
necessarily, inequality (\ref{reallocineq}) holds. To see this, suppose that
$I_{F_0,F_{\sigma_{A_i}}}(z_1)=0=I_{F_0,F_{\sigma_{A_i}}}(z_2)$. For all 
$\mu_i$ such that the corresponding $z_{\mu_i}\in (z_1,z_2)$, since the 
posteriors $\mu_i$ have the same supports as $\mu_{\sigma_{A_i}}$, one must
have that 
$$\mu_i(s\notin [z_1,z_2]\vert z_{\mu_i}\in(z_1,z_2))=0$$
This immediately means that $I_{F_0^i,F_i}(z)=0$ for $z\in \{z_1,z_2\}$. 
So, if they have the same support over posteriors, NIPMC guarantees that 
(\ref{reallocineq}) holds. Thus NIPMC implies that NIPC holds whenever payoffs 
from $u$ depend only on posterior means.

By contrast, when the payoffs and costs of information depend on the 
entire posterior distribution, as in \cite{denti2022posterior}, the 
sequence of distributions over posteriors could satisfy NIPMC, yet still
be suboptimal. This is because the posterior means do not contain enough 
information to determine the optimum when the payoffs depend on the 
entire posterior. Consider, for instance, the following example:

\noindent\textbf{Example 1:} Consider $Z=\{0,0.5,1\}$ with a uniform 
prior, and posterior-separable cost of information acquisition equal to
reduction in entropy, i.e. $H(\pi)=\sum_{z\in Z} \pi(z)\ln(\pi(z))$. 
Let $A_1\coloneqq\{a_1,a_2,a_3\}$ and 
$A_2\coloneqq\{\hat{a}_1,a_2,\hat{a}_3\}$, with
\begin{equation*}
       u(a,z)=\left\{ 
       \begin{array}{ll}
       \ln 550 - \ln 9 - z\ln 25, & 
       a=a_1\\    
       \ln 100 - z\ln 100 ,& 
       a=\tilde{a}_1\\ 
       \ln 22, & a=a_2\\
      \ln 22 - \ln 9 + z\ln 25 ,& a=a_3\\
        z\ln 100 ,& a=\tilde{a}_3\\
       \end{array}
       \right. 
\end{equation*}

From \cite{matvejka2015rational}, equation (13), one can verify that for
both $A_1$ and $A_2$, the optimal distribution of posterior means has 
support on $z\in\{\frac{1}{6},\frac{1}{2},\frac{5}{6}\}$, with 
respective probabilities $\{\frac{3}{8},\frac{1}{4},\frac{3}{8}\}$ and
$\{\frac{55}{114},\frac{2}{57},\frac{55}{114}\}$; note that the MPC 
constraint (\ref{MPC}) does not hold at any $z\in(0,1)$ in either case.
However, they do not share the same optimal posteriors. 

From the perspective of the resultant posteriors, the example satisfies 
NIPC since $H$ is posterior separable. However, when looking at the 
posterior means, one must conclude that NIPMC is not satisfied. 
Consider, then, the case where instead of choosing the above 
probabilities of the posteriors, the probabilities over 
$z\in\{\frac{1}{6},\frac{1}{2},\frac{5}{6}\}$ are redistributed to 
$ \{163/486, 65/228,163/486  \}$ and $ \{ 1/2, 0, 1/2\}$ for $A_1$ and
$A_2$, respectively. Notice that this redistribution is admissible since
the MPC constraint is not binding. However, it is straightforward
to verify that there is no set of posteriors with these posterior 
means for which these distributions are respectively optimal. $\square$

The role that Axiom 2 plays is in providing the necessary conditions for a 
distribution of posteriors to be optimal if the DM has posterior-
mean separable costs of information acquisition. Suppose that 
$F_{\sigma_{A_i}}$ is optimal for each menu $A_i$ given prior $F_0$. Then
if we switch the prior to $F_0^i$, while keeping the menu $A_i$ fixed,
then the resulting distribution of posteriors $F_i$, as described in Axiom
\ref{NIPMC}, will be optimal. We formalize this in the lemma below:

\begin{lemma}
\label{DM2019}
    Fix menu $A$, and suppose that the DM has finite, Lipschitz continuous, 
    posterior-mean separable costs of information generated by $c:[0,1]\rightarrow\mathbb{R}$. 
    Suppose that $F_{\sigma_{A}}$ is an optimal distribution of posterior 
    means given prior $F_0$ (the state space need not be finite). 
    
    Then given prior $\hat{F}_0$, if for $\hat{F}_{\sigma_A}\in\mathcal{I}_{\hat{F}_0}$,
    $\mbox{supp}(\hat{F}_{\sigma_A})\subset \mbox{supp}(F_{\sigma_A})$ 
    and $I_{F_0,F_{\sigma_{A}}}(z)=0\implies I_{\hat{F}_0,\hat{F}_{\sigma_A}}(z)=0$, 
    the distribution of posterior means $\hat{F}_{\sigma_A}$ is optimal.
\end{lemma}

Notice that our axiom, as well as the lemma, differ from the analogous 
Lemma 4 in \cite{denti2022posterior} in that we have an additional
condition regarding the mean-preserving contraction constraint, 
$I_{F_0,F_{\sigma_A}}$. This is due to the difference in optimization 
problems between the case of optimizing over distributions over 
\emph{posteriors}, as Denti does, and optimizing over distributions 
over \emph{posterior means}. As Denti points out, his Lemma 4 takes
advantage of the property of ``locally invariant posteriors" (LIP) of 
\cite{caplin2022rationally}, which states that the support of the optimal 
distribution of posteriors does not change for small perturbations,
assuming this support is still feasible. This property flows from the 
concavification arguments over posteriors that have become standard in 
information design and Bayesian persuasion, most notably by 
\cite{Aumann1995} and \cite{Kamenica2011}. 

It should be recognized that Lemma 4 may be of independent interest in
and of itself, in that it provides a posterior-mean analogue of LIP. That
is, given prior $\hat{F}_0$ that is close to $F_0$, the support of the 
optimal posterior means remains the same, but with the crucial caveat 
that the MPC constraint (\ref{MPC}) holds at the same values. 
This gives a recipe for candidate solutions for the optimal information,
given a solution for optimal distribution of posterior means $F$ for 
prior $F_0$. Take as given the intervals on which (\ref{MPC})
is not binding for $F_0$. For alternative prior $\hat{F}_0$, if there 
is some distribution of posterior means $\hat{F}$ such that, 
interval-by-interval, there is a mean-preserving contraction of $\hat{F}_0$ 
with support restricted to that of $F$, then this mean-preserving 
contraction is also optimal. Thus, in many directions of perturbation (i.e.
those that keep the MPC constraint binding at the same locations), the
same supports of posterior means will remain optimal.

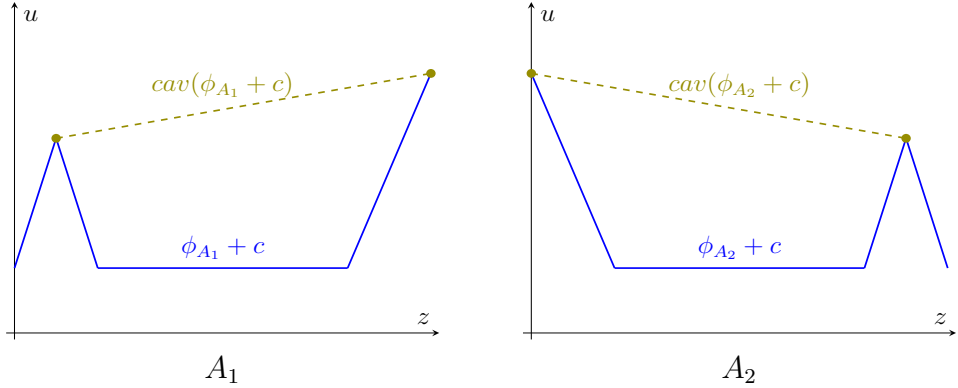
\begin{figure}[h]
\centering
\begin{subfigure}{0.4\textwidth}
\begin{center}
\resizebox{60mm}{45mm}{
\begin{tikzpicture}
\begin{axis}[
	axis lines=center,
  ymin=-0.02,ymax=1.02,
  xmin=-0.02,xmax=1.02,
  xlabel={$z$}, ylabel={$u$},
  ymajorticks=false,
  xmajorticks=false
]
\addplot[blue, thick,samples=200][domain=0:0.1] {4*x+0.2};
\addplot[blue, thick,samples=200][domain=0.1:0.2] {-4*(x-0.1)+0.6};
\addplot[blue, thick,samples=200][domain=0.2:0.8] {0.2};
\addplot[blue, thick,samples=200][domain=0.8:1] {0.2+3*(x-0.8)};
\addplot[olive,thick,dashed,samples=200][domain=0.1:1]{0.6+2/9*(x-0.1)};
\node[fill, olive, circle, inner sep=1.5pt] 	at (axis cs:0.1,0.6){};
\node[fill, olive, circle, inner sep=1.5pt] 	at (axis cs:1,0.8){};
\draw[blue] (axis cs: 0.5, 0.2) node[above] {$\phi_{A_1}+c$};
\draw[olive] (axis cs: 0.5, 0.7) node[above] {$cav(\phi_{A_1}+c)$};
\end{axis}
\end{tikzpicture}
}

$A_1$
\end{center}
\end{subfigure}
\hspace{5 mm}
\begin{subfigure}{0.4\textwidth}
\begin{center}
\resizebox{60mm}{45mm}{
\begin{tikzpicture}
\begin{axis}[
	axis lines=center,
  ymin=-0.02,ymax=1.02,
  xmin=-0.02,xmax=1.02,
  xlabel={$z$}, ylabel={$u$},
  ymajorticks=false,
  xmajorticks=false
]
\addplot[blue, thick,samples=200][domain=0.9:1] {-4*(x-0.9)+0.6};
\addplot[blue, thick,samples=200][domain=0.8:0.9] {4*(x-0.8)+0.2};
\addplot[blue, thick,samples=200][domain=0.2:0.8] {0.2};
\addplot[blue, thick,samples=200][domain=0:0.2] {0.8-3*x};
\addplot[olive,thick,dashed,samples=200][domain=0:0.9]{0.8-2/9*x};
\node[fill, olive, circle, inner sep=1.5pt] 	at (axis cs:0.9,0.6){};
\node[fill, olive, circle, inner sep=1.5pt] 	at (axis cs:0,0.8){};
\draw[blue] (axis cs: 0.5, 0.2) node[above] {$\phi_{A_2}+c$};
\draw[olive] (axis cs: 0.5, 0.7) node[above] {$cav(\phi_{A_2}+c)$};
\end{axis}
\end{tikzpicture}
}

$A_2$
\end{center}
\end{subfigure}

\caption{Optimal information acquisition}
\end{figure}

As shown in Figure 1, this implies that the upper envelope of the values of 
posteriors from the optimal distribution, net of information costs, must all 
lie on the same hyperplane $P$. Any unused posterior must lie weakly below this 
hyperplane (Lemma 3 of \cite{caplin2013behavioral}). Otherwise, as shown in 
Figure 2, there would be an alternative way of distributing posteriors that 
yielded a higher value. For instance, if there were an unused posterior, in the 
convex hull of the support of the distribution of posteriors, that had higher 
net value, a mean-preserving contraction of this distribution would improve the 
DM's (ex-ante) expected payoff. Conversely, if there were more ``extreme" 
posteriors with higher net value, a mean-preserving spread would increase the 
DM's expected payoff. Thus, a reallocation away from the given posteriors
would not improve the DM's payoff.

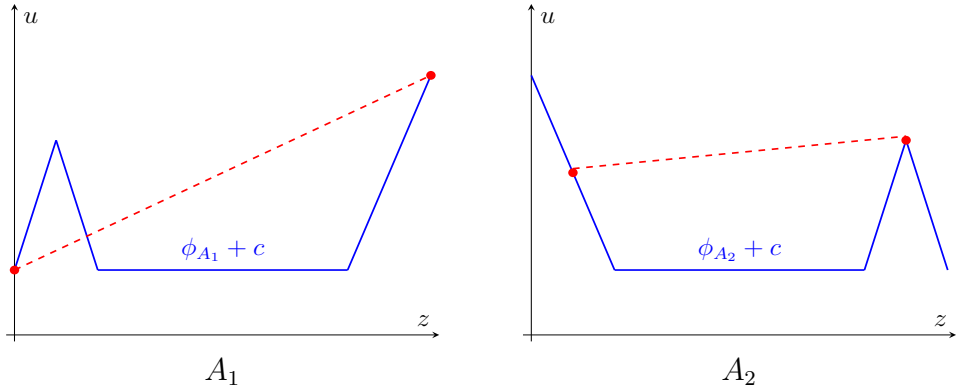
\begin{figure}[h]
\centering
\begin{subfigure}{0.4\textwidth}
\begin{center}
\resizebox{60mm}{45mm}{
\begin{tikzpicture}
\begin{axis}[
	axis lines=center,
  ymin=-0.02,ymax=1.02,
  xmin=-0.02,xmax=1.02,
  xlabel={$z$}, ylabel={$u$},
  ymajorticks=false,
  xmajorticks=false
]
\addplot[blue, thick,samples=200][domain=0:0.1] {4*x+0.2};
\addplot[blue, thick,samples=200][domain=0.1:0.2] {-4*(x-0.1)+0.6};
\addplot[blue, thick,samples=200][domain=0.2:0.8] {0.2};
\addplot[blue, thick,samples=200][domain=0.8:1] {0.2+3*(x-0.8)};
\addplot[red,thick,dashed,samples=200][domain=0:1]{0.2+0.6*x)};
\node[fill, red, circle, inner sep=1.5pt] 	at (axis cs:0,0.2){};
\node[fill, red, circle, inner sep=1.5pt] 	at (axis cs:1,0.8){};
\draw[blue] (axis cs: 0.5, 0.2) node[above] {$\phi_{A_1}+c$};
\end{axis}
\end{tikzpicture}
}

$A_1$
\end{center}
\end{subfigure}
\hspace{5 mm}
\begin{subfigure}{0.4\textwidth}
\begin{center}
\resizebox{60mm}{45mm}{
\begin{tikzpicture}
\begin{axis}[
	axis lines=center,
  ymin=-0.02,ymax=1.02,
  xmin=-0.02,xmax=1.02,
  xlabel={$z$}, ylabel={$u$},
  ymajorticks=false,
  xmajorticks=false
]
\addplot[blue, thick,samples=200][domain=0.9:1] {-4*(x-0.9)+0.6};
\addplot[blue, thick,samples=200][domain=0.8:0.9] {4*(x-0.8)+0.2};
\addplot[blue, thick,samples=200][domain=0.2:0.8] {0.2};
\addplot[blue, thick,samples=200][domain=0:0.2] {0.8-3*x};
\addplot[red,thick,dashed,samples=200][domain=0.1:0.9]{0.5+1/8*x};
\node[fill, red, circle, inner sep=1.5pt] 	at (axis cs:0.9,0.6){};
\node[fill, red, circle, inner sep=1.5pt] 	at (axis cs:0.1,0.5){};
\draw[blue] (axis cs: 0.5, 0.2) node[above] {$\phi_{A_2}+c$};
\end{axis}
\end{tikzpicture}
}

$A_2$
\end{center}
\end{subfigure}

\caption{Reallocation of posteriors does not improve}
\end{figure}

By contrast, our environment is one of information design over posterior
means. As a result, rather than using concavification arguments, the 
relevant characterization is by a \emph{price function} $P$ as defined 
by \cite{dworczak2019simple}. Prima facie, the set of posterior means
is one-dimensional, and so one might think the information acquisition 
problem would be simpler than in the potentially multidimensional 
probability simplex for all posteriors. However, an attempt to simply replace
``posteriors" with ``posterior means" in Denti's NIPC condition does not
work, as illustrated in the following example.

\textbf{Example 3:} Suppose that, as in Figure 3,
\[
F_0(z)=\begin{cases}
0.49, & z\in[0,0.4)\\
0.5, & z\in[0.4,0.6)\\
0.51, & z\in[0.6,1)\\
1, & z=1
\end{cases}
\]
\[
\phi_{A_1}(z)+c(z)=\begin{cases}
2-\frac{10}{3}z, & z\in[0,0.3)\\
1+0.1(z-0.3), & z\in [0.3,0.5)\\
1.02-0.1(z-0.5), & z\in[0.5,0.7)\\
1+\frac{10}{3}(z-0.7), & z\in[0.7,1]
\end{cases}
\]

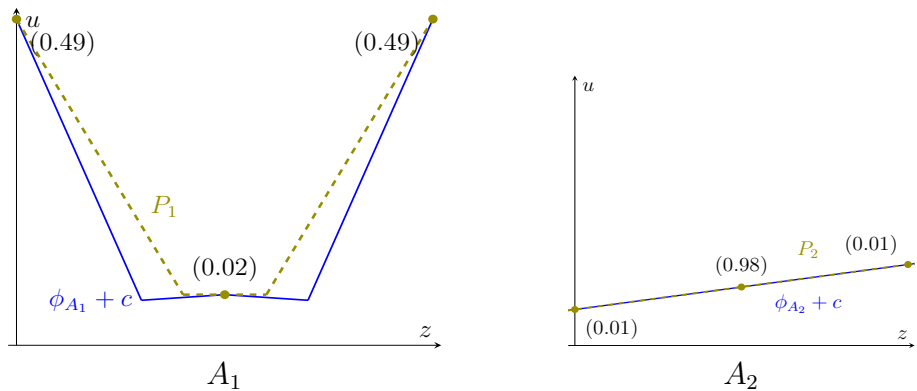
\begin{figure}[h]
\centering
\begin{subfigure}{0.4\textwidth}
\begin{center}
\resizebox{60mm}{45mm}{
\begin{tikzpicture}
\begin{axis}[
	axis lines=center,
  ymin=0.84,ymax=2.04,
  xmin=-0.02,xmax=1.02,
  xlabel={$z$}, ylabel={$u$},
  ymajorticks=false,
  xmajorticks=false
]
\addplot[blue, thick,samples=200][domain=0:0.3] {2-10/3*x};
\addplot[blue, thick,samples=200][domain=0.3:0.5] {1+0.1*(x-0.3)};
\addplot[blue, thick,samples=200][domain=0.5:0.7] {1.02-0.1*(x-0.5)};
\addplot[blue, thick,samples=200][domain=0.7:1] {1+10/3*(x-0.7)};
\addplot[olive,very thick,dashed,samples=200][domain=0:0.4]{2-2.45*x};
\addplot[olive,very thick,dashed,samples=200][domain=0.4:0.6]{1.02};
\addplot[olive,very thick,dashed,samples=200][domain=0.6:1]{1.02+2.45*(x-0.6)};
\node[label = {south east: (0.49)}, fill, olive, circle, inner sep=1.5pt] 	at (axis cs:0,2){};
\node[label = {south west: (0.49)}, fill, olive, circle, inner sep=1.5pt] 	at (axis cs:1,2){};
\node[label = {(0.02)}, fill, olive, circle, inner sep=1.5pt] 	at (axis cs:0.5,1.02){};
\draw[blue] (axis cs: 0.3, 1) node[left] {$\phi_{A_1}+c$};
\draw[olive] (axis cs: 0.3, 1.333) node[right] {$P_1$};
\end{axis}
\end{tikzpicture}
}

$A_1$
\end{center}
\end{subfigure}
\hspace{5 mm}
\begin{subfigure}{0.4\textwidth}
\begin{center}
\resizebox{48 mm}{36mm}{
\begin{tikzpicture}
\begin{axis}[
	axis lines=center,
  ymin=0.84,ymax=2.04,
  xmin=-0.02,xmax=1.02,
  xlabel={$z$}, ylabel={$u$},
  ymajorticks=false,
  xmajorticks=false
]
\addplot[blue, thick,samples=200] {1+0.2*x};
\addplot[olive,very thick,dashed,samples=200] {1+0.2*x};
\node[label={south east:(0.01)}, fill, olive, circle, inner sep=1.5pt] 	at (axis cs:0,1){};
\node[label={(0.98)}, fill, olive, circle, inner sep=1.5pt] 	at (axis cs:0.5,1.1){};
\node[label={north west:(0.01)}, fill, olive, circle, inner sep=1.5pt] 	at (axis cs:1,1.2){};
\draw[blue] (axis cs: 0.7, 1.1) node[below] {$\phi_{A_2}+c$};
\draw[olive] (axis cs: 0.7, 1.2) node[above] {$P_2$};
\end{axis}
\end{tikzpicture}
}

$A_2$
\end{center}
\end{subfigure}

\caption{Optimal information acquisition with $F_0$ as in Example 2}
\end{figure}

Consider now the alternative priors $\hat{F}_0$ which place probabilities 
of $0.5$ on $z=0$ and $z=1$, respectively, and the reallocation of 
posterior means from Figure 3 that yields $\hat{F}_1=\delta_{0.5}$, and 
$\hat{F}_2=\hat{F}_0$ (Figure 4). Without condition (\ref{allineq}),  
NIPMC would be the same as Denti's NIPC, with ``posteriors" replaced with 
``posterior means;" however, as can clearly be seen, this would be 
suboptimal: rather, the optimal concavification in the left-hand graph is 
also $\hat{F}_1=\hat{F}_0$, which is feasible when the prior only has 
support over the two states $z\in\{0,1\}$.

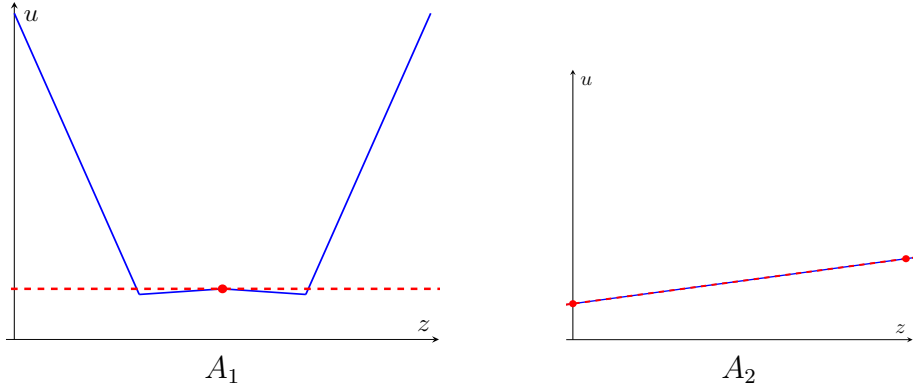
\begin{figure}[h]
\centering
\begin{subfigure}{0.4\textwidth}
\begin{center}
\resizebox{60mm}{45mm}{
\begin{tikzpicture}
\begin{axis}[
	axis lines=center,
  ymin=0.84,ymax=2.04,
  xmin=-0.02,xmax=1.02,
  xlabel={$z$}, ylabel={$u$},
  ymajorticks=false,
  xmajorticks=false
]
\addplot[blue, thick,samples=200][domain=0:0.3] {2-10/3*x};
\addplot[blue, thick,samples=200][domain=0.3:0.5] {1+0.1*(x-0.3)};
\addplot[blue, thick,samples=200][domain=0.5:0.7] {1.02-0.1*(x-0.5)};
\addplot[blue, thick,samples=200][domain=0.7:1] {1+10/3*(x-0.7)};
\addplot[red,very thick,dashed,samples=200] {1.02};
\node[fill, red, circle, inner sep=1.5pt] 	at (axis cs:0.5,1.02){};
\end{axis}
\end{tikzpicture}
}

$A_1$
\end{center}
\end{subfigure}
\hspace{5 mm}
\begin{subfigure}{0.4\textwidth}
\begin{center}
\resizebox{48 mm}{36mm}{
\begin{tikzpicture}
\begin{axis}[
	axis lines=center,
  ymin=0.84,ymax=2.04,
  xmin=-0.02,xmax=1.02,
  xlabel={$z$}, ylabel={$u$},
  ymajorticks=false,
  xmajorticks=false
]
\addplot[blue, thick,samples=200] {1+0.2*x};
\addplot[red,very thick,dashed,samples=200] {1+0.2*x};
\node[fill, red, circle, inner sep=1.5pt] 	at (axis cs:0,1){};
\node[fill, red, circle, inner sep=1.5pt] 	at (axis cs:1,1.2){};
\end{axis}
\end{tikzpicture}
}

$A_2$
\end{center}
\end{subfigure}

\caption{Suboptimal information acquisition after reallocation, with supp$(\hat{F}_0)=\{0,1\}$}
\end{figure}

The reason that adapting NIPC in this way fails is that it is no longer the
case that all mean-preserving spreads of the prior mean on this interval
are possible: they are subject to the constraint given by (\ref{feasiblecdfs}). 
This means that, given a CDF $F\in\mathcal{I}_{F_0}$, mean-preserving 
contractions of $F$ are feasible, while mean-preserving spreads may not 
be since they may not be in $\mathcal{I}_{F_0}$. As a result, there cannot 
be any concave regions of the function $P$ that majorizes $\phi_A+c$, as if 
so a mean-preserving contraction would take advantage of that. By contrast, 
there may be \emph{convex} regions of $P$, since the mean-preserving spread to 
take advantage of it may not be feasible. Instead of being
majorized by a hyperplane $P$ (here, an affine function) as in information 
designs over posteriors, the DM's objective is majorized by a convex
function. So, in order to preserve optimality, the convex majorizing function
must be preserved as well.

The relationship of Denti's NIPC axiom and our NIPMC axiom can be summarized as 
follows through the respective acyclicality conditions. NIPC states that the 
SDSC data can be viewed as coming from optimal \emph{concavification}, with 
reallocations of posteriors yielding suboptimality. By contrast, NIPMC states
that the SDSC is consistent with optimal \emph{majorization by a convex 
function}.

The characterization of information acquisition as restricted by 
(\ref{feasiblecdfs}) is what drives the necessity of Axiom \ref{NIPMC}.
Thus, when the cost of information is posterior-mean separable, 
Axiom \ref{NIPMC} is satisfied for $F_0^i=F_0$. By Lemma \ref{DM2019}, 
for each $i$, we have 
\[
\int [\phi_{A_i}(z)+c(z)]dF_i(z)\geq \int [\phi_{A_i}(z)+c(z)]dG_i(z)
\]
Taking the weighted average via $\beta_i$,
\[
\sum_{i=1}^n \beta_i\int[ \phi_{A_i}(z)+c(z)]dF_i(z)\geq \sum_{i=1}^n \beta_i\int [\phi_{A_i}(z)+c(z)]dG_i(z)
\]
Since we assumed that $\{G_i\}_{i=1}^n$ is a reallocation of posterior
means, we have that 
\[
\sum_{i=1}^n \beta_i\int c(z)dF_i(z)= \sum_{i=1}^n \beta_i\int c(z)dG_i(z)
\]
and therefore 
\[
\sum_{i=1}^n \beta_i\int \phi_{A_i}(z)dF_i(z)\geq \sum_{i=1}^n \beta_i\int \phi_{A_i}(z)dG_i(z)
\]
as required. This will be useful in our construction of the representation
of the cost function for fixed prior $F_0$.

This brings us to our main result, which shows that Axioms \ref{NIAS} and
\ref{NIPMC} are precisely what we need for our characterization.

\begin{theorem}
\label{costexist}
    For a given $u$ and $F_0$, the dataset $\{\sigma_{A_i}\}_{i=1}^n$ 
    satisfies NIAS and NIPMC if and only if it can be represented as 
    coming from a DM with posterior-mean separable costs of information
    for which $\{(F_{\sigma_{A_i}},D_{\sigma_{A_i}})\}_{i=1}^n$ are
    optimal information acquisitions and decision functions, respectively.
\end{theorem}

The overall structure of the argument is similar to that in 
\cite{denti2022posterior}, though the details of the particular steps differ
significantly. In the ``necessity" direction, we show how to normalize the 
potentially different prior distributions $F_0^i$ to come from the same prior, 
and then follow a similar argument to that in the discussion above Theorem 
\ref{costexist}. In the ``sufficiency" direction, we translate the problem 
into a convex program over $S_A$, and show that Axiom \ref{NIPMC} is 
equivalent to a system of linear inequalities which satisfy a version of 
``cyclical monotonicity," introduced by 
\cite{rockafellar1966characterization}, and first used in revealed preference 
analysis by \cite{afriat1967construction}, and in mechanism design by 
\cite{rochet1987necessary}. Recall that \cite{dworczak2019simple} show that there 
exists some price function $P$ that majorizes the DM's utility as a function of 
posterior mean, here being $\max_{a\in A_i} u(a,z)+c(z)$. Using Farkas' lemma,%
\footnote{This approach is also used in the previous literature on identifying
information costs \citep{caplin2015revealed,chambers2020costly,denti2022posterior}.}
we show the existence of a vector $\lambda$ which can be used to both 
construct $c$ as well as the price function $P$ that rationalize $\sigma_{A_i}$.

Given the existence of a representation of information costs via a 
posterior-mean separable function $c$, one can use the construction of
$c$ in the above result to identify such a possible representation. In 
particular, the aforementioned vector $\lambda$ provides a set of linear 
inequalities which the payoffs must satisfy, given by 
(\ref{optimuminequality}) in the proof. One can then identify the set of such 
$\lambda$ that satisfy these inequalities, and then construct $c$ as in 
equation (\ref{IDcost}).

It should be noted that we do not restrict the function $c$ in Theorem 
\ref{costexist} to be concave, and hence it may be non-canonical (though the 
case of $c$ being concave is also covered by our theorem). We discuss the 
property of concavity of $c$ in Section 4.6.

\section{Extensions and Discussion}

\subsection{Uniform Posterior-Mean Separability}

Up to now, we have focused on the case of a DM who must make decisions
with respect to a \emph{single} prior $F_0$. Thus our representation 
result only holds with respect to $F_0$; it does not say anything about
what the representation would be if the prior were to change. One may also
be interested in the case where the cost function $c$ is \emph{independent}
of the prior; thus, for \emph{all} $F_0$, the cost of information acquisition
would be given by (\ref{pmsepcosts}). In other words, analogous to the 
``uniform posterior separability" (UPS) condition of 
\cite{caplin2022rationally}, where the function $H$ in (\ref{postsepcosts})
does not depend on the prior $\pi$, one would here have a notion of
``uniform posterior-mean separability" (UPMS).

\begin{defn}
    A cost function $C:\mathcal{I}_{F_0}\rightarrow(-\infty,\infty)$ is 
    \emph{uniformly posterior-mean separable} if, for all 
    $F_0\in \mathcal{F}$ (with respective prior means $\mu_0$,
    there is an upper-semicontinuous function $c:[0,1]\rightarrow \mathbb{R}$ 
    such that 
    \begin{equation}
        C(F)=c(z_0)-\int_0^1 c(z)dF(z)
    \end{equation}
\end{defn}

The question, then, is what properties guarantee that $c$ is the same for
all $F_0$, assuming that for each $F_0$, Axioms \ref{NIAS} and \ref{NIPMC}
hold. \cite{denti2022posterior} argues that in the case of posterior-separable
costs, the extension to uniform posterior separability is made by a 
straightforward strengthening of NIPMC, which allows for different
priors for each decision problem, and then reallocating across them. This 
turns out to be the case in our environment as well.

For each $i\in \{1,...,n\}$, let the respective priors be $\tilde{F}_0^i$,
which we assume each have full support in $Z$. We modify the subscripts in the 
definitions of previous variables and functions that had $A$ to now include 
$\tilde{F}_0$ as well, and replace the prior $F_0$ with $\tilde{F}_0^i$ as
needed.

We now state the appropriate axioms, the first simply being a rewriting of
NIAS with the modified notation, and the second being the aforementioned
strengthening of NIPMC.

\begin{axiom}
    \label{UNIAS}
    For every $A\in\mathcal{A}$ with associated prior $\tilde{F}_0$, every 
    $a\in\mbox{supp}(\sigma_A)$, and $b\in A$,
    \[
        u(a,z_{\sigma_{(\tilde{F}_0,A)}}(a))\geq u(b,z_{\sigma_{(\tilde{F}_0,A)}}(a))
    \]
\end{axiom}

\begin{axiom}
    \label{UNIPMC}    
    For $i\in\{1,...,n\}$, consider the sequence of menus
    $A_i\in\mathcal{A}$, coefficients $\beta_i\in \mathbb{R}_+$, priors 
    $F_0^i,\tilde{F}_0^i\in \mathcal{F}$ such that 
    $\mbox{supp}(F_0^i)\subset Z$, and revealed distributions of posterior 
    means $F_i\in \mathcal{I}_{F_0^i}$ such that:
    \begin{enumerate}[(i)]
        \item $\mbox{supp}(F_i)\subset \mbox{supp}(F_{\sigma_{(\tilde{F}_0^i,A_i)}})$, and
        \item For all $z\in [0,1]$, 
        \[
            I_{\tilde{F}_0^i,F_{\sigma_{(\tilde{F}_0^i,A_i)}}}(z)=0\implies I_{F^i_0,F_i}(z)=0
        \]
    \end{enumerate} 
    Then for every reallocation of posterior means $\{G_i\}_{i=1}^n$,
    \[
        \sum_{i=1}^n \beta_i \int \phi_{A_i}(z)dF_i(z)\geq \sum_{i=1}^n \beta_i \int \phi_{A_i}(z)dG_i(z)
    \]
\end{axiom}

\begin{theorem}
    For a given $u$ and $\{\tilde{F}_0^i\}_{i=1}^N$, the dataset 
    $\{\sigma_{\tilde{F}_i^0,A_i}\}_{i=1}^n$ satisfies Axioms \ref{UNIAS} and
    \ref{UNIPMC} if and only if it can be represented as 
    coming from a DM with \emph{uniformly} posterior-mean separable costs of 
    information for which $\{ (F_{\sigma_{\tilde{F}_0^i,A_i}}, D_{\sigma_{\tilde{F}_0^i,A_i}})\}_{i=1}^n$ are the
    optimal information acquisitions and decision functions, respectively.
\end{theorem}

As in \cite{denti2022posterior}, the proof is substantially identical to 
the case of a single prior, and so is omitted.

\subsection{Continuum of States/Acts}

The analysis in Theorem \ref{costexist} pertains to the case of finite 
states and acts. However, a large part of the motivation for the
use of models of information design where the problem only depends on
the distribution of posterior means is to allow for a continuum of 
states and acts \citep{dworczak2019simple}. Indeed, many of the 
natural applications for these models of information acquisition also
take advantage of techniques stemming from the continuum of states
\citep{mensch2022monopoly}. The question of extension to a continuum
is thus important both in terms of motivation and application.

Our model is modified as follows; where unmodified, the same notation
and terminology is used as before. Let $X$ be a compact grand set of acts,
and endow it with the Borel $\sigma$-algebra. Let $\mathcal{A}$ be a 
set of compact subsets of $X$, with no redundant/dominated acts: that
is, for all $a,a^\prime\in A\subset \mathcal{A}$, 
$u(a,\cdot)\neq u(a^\prime,\cdot)$, and $\exists z,z^\prime\in[0,1]$ such
that $u(a,z)> u(a^\prime,z)$ and $u(a,z^\prime)<u(a^\prime,z^\prime)$.
For each $A\in\mathcal{A}$, the analyst observes the measurable function 
$\sigma_A:Z\rightarrow\Delta(A)$. We can then analogously define the decision 
function $D_A:[0,1]\rightarrow\Delta(A)$ as in the finite model, namely, for 
all measurable $B\subset A$,
\[
\int_0^1 \int_B d\sigma_A(a\vert z)dF_0(z)=\int_0^1 \int_B dD_A(a\vert \tilde{z})dF(\tilde{z})
\]
Let $\mathcal{D}_A$ be the set of measurable functions from $Z$ to $\Delta(A)$. 
Then $(F,D_A)$ is optimal for the DM if 
\[
\int_{Z}\int_A u(a,z)dD_A(a\vert z)dF(z)-C(F)\geq \int_Z \phi_A(z)d\hat{F}(z)-C(\hat{F}),\forall \hat{F}\in\mathcal{I}_{F_0}
\]
Define 
\[
\sigma_A(B)=\int_0^1 \sigma_A(B\vert z)dF_0(z)
\]
as the unconditional probability that $a\in B$ is chosen according to $\sigma_A$.
For each measurable $B\subset \mbox{supp}(\sigma_A)$ such that $\sigma_A(B)>0$, 
the \emph{revealed posterior mean} $z_{\sigma_A}(B)$ is given by 
\[
z_{\sigma_A}(B)=\frac{\int_0^1\int_B zd\sigma_A(a\vert z)dF_0(z)}{\int_0^1\int_B d\sigma_A(a\vert z)dF_0(z)}
\]
For $B$ such that $B\cap\mbox{supp}(\sigma_A)=\emptyset$, adopt the 
convention that $z_{\sigma_A}(B)=z_0$. We can extend the above definition to 
sets $B$ of measure $0$ by the use of the Lebesgue differentiation theorem. 
Thus for almost all $a\in A$ with respect to $\sigma_A$, the revealed posterior 
mean $z_{\sigma_A}(a)$ is well-defined. The revealed decision function  
$D_{\sigma_A}\in\mathcal{D}_A$ is defined as follows:
\begin{enumerate}[(i)]
    \item If $z\in\mbox{supp}(F_{\sigma_A})$ and $z=z_{\sigma_A}(a)$, then 
    for all $Y$ such that $z\in Y$,
    \[
    \int_Y dD_{\sigma_A}(a\vert z)dF_{\sigma_A}(z)=d\sigma_A(a)
    \]
    \item If $z\in\mbox{supp}(F_{\sigma_A})$ and $z\neq z_{\sigma_A}(a)$,
    then $dD_{\sigma_A}(a\vert z)=0$;
    \item If $z\notin\mbox{supp}(F_{\sigma_A})$, then 
    $D_{\sigma_A}(\cdot\vert z)=\sigma_A(\cdot)$.
\end{enumerate}
In particular, if each act $a$ has a distinct revealed posterior mean, then 
$dF_{\sigma_A}(z_{\sigma_A}(a))=d\sigma_A(a)$ and $D_{\sigma_A}(a\vert z_{\sigma_A}(a))=1$.

The adaptations of Lemmas \ref{sufficiency} and \ref{revealedoptimal} is 
straightforward and therefore omitted.

\begin{theorem}
    \label{continuumexist} The SDSC with $Z\subset[0,1]$ and $A\subset 
    \mathcal{A}$ compact satisfies Axioms \ref{NIAS} and 
    \ref{NIPMC} if and only if it can be rationalized by a posterior-mean
    separable cost of information acquisition.
\end{theorem}

\subsection{Monotone Partitional Information}

A special case of interest is where the information chosen by the DM,
as represented by the revealed posterior means, is \emph{monotone 
partitional}: one can divide the set of states into intervals, within 
which the same act is chosen (with possible mixing between the two 
acts on the boundary). We present a formal definition of a monotone
partitional information structure below.

\begin{defn}
Let $F\in \mathcal{I}_{F_0}$. $F$ is \emph{monotone partitional} if,
for all $z_1,z_2\in \mbox{supp}(F)$, where $z_1<z_2$, there exists 
$z\in[z_1,z_2]$ such that $I_{F_0,F}(z)=0$.
\end{defn}

An important implication of the monotone partitional structure is that 
the mean-preserving contraction constraint (\ref{MPC}) must hold at the
endpoints of these intervals, and that there is exactly one posterior mean
in each such interval. These observations allow for a more straightforward 
expression of Axiom \ref{NIPMC} whenever $F_{\sigma_{A_i}}$ is monotone 
partitional for each $i\in\{1,...,n\}$.

In the case where the prior $F_0$ is absolutely continuous over $[0,1]$
(see Section 5.2), a monotone partitional structure can be expressed by
an interval partition of $[0,1]$, and taking the conditional mean of $z$ 
for each element of the partition under $F_0$. Axiom \ref{NIPMC} is then 
equivalent to stating that, for every sequence of priors $\{F_0^i\}$
with the same conditional means under the same partitional structure, 
there does not exist a reallocation of posterior means that improves the
DM's payoff.

In the case where the prior $F_0$ is not absolutely continuous, the same
idea can be expressed via  the probability integral transform, i.e. a 
distribution of $\alpha$ uniformly distributed over $[0,1]$, where 
$\chi(\alpha)=\inf\{z:F_0(z)\geq \alpha\}$. A monotone partitional 
structure can then be represented as before by a partition of $[0,1]$ and 
taking the conditional mean of $\chi(\alpha)$ within each element of the 
partition.%
\footnote{To avoid any pathologies at the endpoints of elements of
partitions regarding measure-$0$ events, we assume that if an element has 
positive length with endpoints $\alpha_1\neq\alpha_2$, then if it includes 
the endpoint $\alpha_1$, we have 
$\chi(\alpha_1)=\lim_{\alpha\rightarrow\alpha_1^+} \chi(\alpha)$, and
similarly for $\alpha_2$.}
Now consider a sequence of priors $\{\hat{F}_0^i\}_{i=1}^n$ over 
$\alpha$ on $[0,1]$, holding the definition of $\chi(\cdot)$ the same as 
before. Axiom \ref{NIPMC} is then equivalent to stating that, if applying 
the same partitional structure yields the same conditional mean of 
$\chi(\alpha)$ for each $i$ given $\hat{F}_0^i$ on each element of the 
partition with positive probability, then there is no reallocation of 
posterior means that improves the DM's payoff.

\subsection{Binary Actions}

Another set of special cases of interest is when the DM faces menus that 
only contain at most $2$ actions each. In these cases, for each $A_i$, 
either (i) the information constraint (\ref{MPC}) never binds with
respect to the distribution of revealed posterior means on $(0,1)$, 
or (ii) the information is monotone partitional. In the former case, one can 
drop the requirement of (\ref{allineq}), and so the requirement reduces
to that of Denti's NIPC, with posterior means instead of posteriors. In the
latter case, one can check the conditions as described in Section 5.3: 
namely, checking for each of the (two) elements of the partition that the
conditional posterior mean remains the same when comparing $F_i$ and 
$F_{\sigma_{A_i}}$. In both cases, the condition of NIPMC simplifies.

\subsection{Comparative Statics}

One may wonder how increased incentives to acquire information affect 
the distribution of posterior means. Such incentives arise when the 
indirect utility function $\phi$ becomes more convex (which in turn occurs
if the set of possible actions is enlarged). In 
\cite{denti2022posterior}, Proposition 3, he shows that beliefs must 
become more extreme: that is, beliefs under more extreme incentives 
cannot lie in the interior of the convex hull of those under the less
extreme incentives. This is because, if information (i.e. a mean-
preserving spread) is valuable under weaker incentives, then it is 
valuable under stronger incentives as well. With posterior-mean separable 
costs, though, this outcome is challenged by the presence of the 
mean-preserving contraction constraint. A more convex utility function
may change the incentives to have this constraint binding in certain 
locations, thereby preventing mean-preserving spreads away from posterior 
means in the relative interior of this convex hull (here, an open 
interval).

Nevertheless, one can still make claims conditional on the points where
this constraint binds. If, indeed, the constraint is not binding on a 
given interval, then stronger incentives rule out the location of
posterior means in this interval. Since the mean-preserving spread is
now feasible, it would be suboptimal to have a posterior mean in this 
interval.

\begin{proposition}
\label{cs}
Suppose the dataset is rationalized by a posterior-mean separable cost. 
Let $A_1$ and $A_2$ be a pair of menus such that $\phi_{A_1}-\phi_{A_2}$ 
is convex at each $z\in[0,1]$. Let $z_1, z_2\in supp (F_{\sigma_{A_2}})$, 
$\mathcal{I}_{F_0,F_{\sigma_{A_1}}}(z) >0$ for each $z\in(z_1,z_2)$, and 
$\phi_{A_1}-\phi_{A_2}$ is strictly convex for some $z \in (z_1, z_2)$.
Then
$$ supp(F_{\sigma_{A_1}}) \cap (z_1, z_2)= \emptyset$$
\end{proposition}

\subsection{Concavity of $c$}

As mentioned earlier, $c$ may not be concave, and thus not monotone in the 
Blackwell order. This may be surprising to some readers, as it has been 
shown (\cite{de2017rationally}, Theorem 2) that one can replace a 
non-monotone cost function with a monotone one, since one can 
always just acquire more information if it weakly decreases the cost. In 
the case of posterior-separable costs of information, this, in turn, 
implies that it is without loss to consider concave functions 
$H:\Delta(\Theta)\rightarrow \mathbb{R}$; see, for instance, 
\cite{lipnowski2022predicting}, Proposition 7, which uses the analogous
argument that one can take a mean-preserving spread of any posterior in a 
region where $H$ is not concave. Indeed, the cost function as derived in 
\cite{denti2022posterior} is concave. Yet when it comes to posterior-%
mean separable costs of information, the crucial difference is that not 
all mean-preserving spreads of a particular posterior mean $z\in(0,1)$ 
are feasible: it will depend on the overall distribution $F$, and 
whether the information constraint is binding. Notice that, 
while any posterior-mean separable cost of information is also posterior 
separable, applying the argument of \cite{lipnowski2022predicting} only 
implies that, for any revealed posterior mean $z_{\sigma_A}(a)$ and 
$\alpha\in[0,1]$,
\[
\alpha c(z_1)+(1-\alpha)c(z_2)\leq c(z_{\sigma_A}(a))
\]
\[
\alpha z_1+(1-\alpha)z_2=z_{\sigma_A}(a)
\]
whenever $I_{F_0,F_{\sigma_A}}(z)>0$ for all $z\in [z_1,z_2]$. Thus one 
cannot automatically replace $c$ with its concavification and preserve the 
same incentives to acquire information.

To illustrate this, consider the following example. 

\textbf{Example 4:} Let $A=\{a_1,a_2,a_3\}$, with 
\[
u(a_1,z)=-\frac{1}{2}z+\frac{1}{4}
\]
\[
u(a_2,z)=\frac{1}{8}
\]
\[
u(a_3,z)=\frac{1}{2}z-\frac{1}{4}
\]
and let 
\[
c(z)=\begin{cases}
\frac{1}{6}z-\frac{1}{36}, & z\leq \frac{1}{6}\\
5-30z, & z\in[\frac{1}{6},\frac{1}{2}]\\
-25+30z, & z\in[\frac{1}{2},\frac{5}{6}]\\
-\frac{1}{6}z+\frac{5}{36}, & z\geq \frac{5}{6}
\end{cases}
\]
Let $Z=\{0,\frac{1}{3},\frac{2}{3},1\}$ with uniform probability $F_0$. 
It is quite clear that the optimal information acquisition is to have 
two posterior beliefs at $\frac{1}{6}$ and $\frac{5}{6}$, each with 
probability $\frac{1}{2}$ (Figure 5). 
\begin{center}

\begin{figure}
\centering
\resizebox{80 mm}{60 mm}{
\begin{tikzpicture}
\begin{axis}[
		axis lines=center,
  ymin=-10.02,ymax=1.02,
  xmin=-0.02,xmax=1.02,
  xlabel={$z$}, ylabel={$u$},
  ymajorticks=false,
  xmajorticks=false
]
\addplot[blue, thick,samples=200][domain=0:0.16667] {-1/2*x+1/4+1/6*x-1/36};
\addplot[blue, thick,samples=200][domain=0.16667:0.25] {-1/2*x+1/4+5-30*x};
\addplot[blue, thick,samples=200][domain=0.25:0.5] {1/8+5-30*x};
\addplot[blue, thick,samples=200][domain=0.5:0.75] {1/8-25+30*x};
\addplot[blue, thick,samples=200][domain=0.75:0.83333] {1/2*x-1/4-25+30*x};
\addplot[blue, thick,samples=200][domain=0.83333:1] {1/2*x-1/4-1/6*x+5/36};
\addplot[olive, dashed, very thick,samples=200][domain=0:0.33333] {3/4-15/4*x};
\addplot[olive, dashed, very thick,samples=200][domain=0.33333:0.66667] {-1/2};
\addplot[olive, dashed, very thick,samples=200][domain=0.66667:1] {3/4+15/4*(x-1)};
\node[fill, olive, circle, inner sep=1.5pt] 	at (axis cs:0.16667,0.125){};
\node[fill, olive, circle, inner sep=1.5pt] 	at (axis cs:0.83333,0.125){};
\end{axis}
\end{tikzpicture}
}

\caption{Optimal information with original costs}
\end{figure}
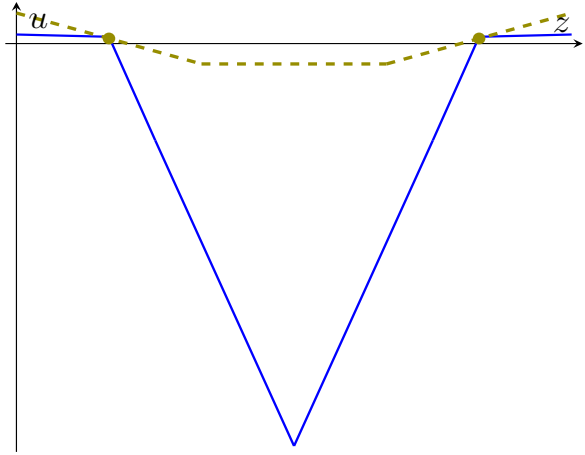
\end{center}

By contrast, the concavification of $c$ is 
given by 
\[
\tilde{c}(z)=\begin{cases}
\frac{1}{6}z-\frac{1}{36}, & z\leq \frac{1}{6}\\
0, & z\in[\frac{1}{6},\frac{5}{6}]\\
-\frac{1}{6}z+\frac{5}{36}, & z\geq \frac{5}{6}
\end{cases}
\]
in which case the optimal information acquisition is to fully reveal if 
$z\in\{0,1\}$, and pool if $z\in\{\frac{1}{3},\frac{2}{3}\}$ (Figure 6). 
To see this, note that the $y$-intercept of $u+\tilde{c}$ is now 
$\frac{2}{9}$, and the slope is $-\frac{1}{3}$. So, the value of $z$ at which 
$\frac{2}{9}-\frac{1}{3}z=\frac{1}{8}$ is $z=\frac{7}{24}$. Therefore, 
there exists $P^*_A$ that is convex and weakly above $u+\tilde{c}$, 
linear on all intervals between $z_1,z_2\in Z$, and is equal to 
$u(z)+\tilde{c}(z)$ for $z\in Z$, given by
$$P^*_A(z)=\begin{cases}
\frac{2}{9}-\frac{7}{24}z, & z\leq \frac{1}{3}\\
\frac{1}{8}, & z\in[\frac{1}{3},\frac{2}{3}]\\
-\frac{5}{72}+\frac{7}{24}z, & z\geq \frac{2}{3}
\end{cases}
$$
This yields strictly higher utility than the original information 
acquisition would, as $P^*_A(\frac{1}{6})>u(\frac{1}{6})+\tilde{c}(\frac{1}{6})$.

\begin{center}
    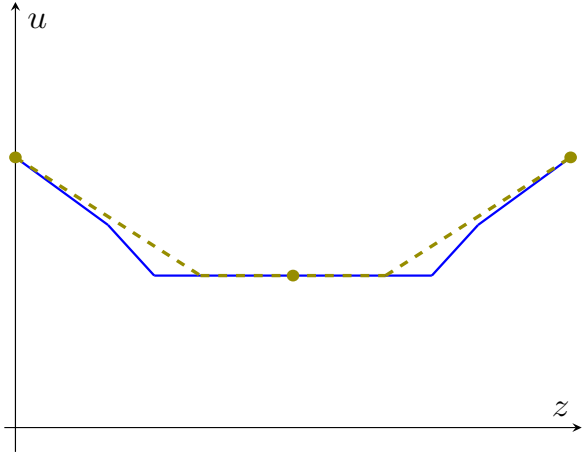
\begin{figure}
\centering
\resizebox{80 mm}{60 mm}{
\begin{tikzpicture}
\begin{axis}[
	axis lines=center,
  ymin=-0.02,ymax=0.35,
  xmin=-0.02,xmax=1.02,
  xlabel={$z$}, ylabel={$u$},
  ymajorticks=false,
  xmajorticks=false
]
\addplot[blue, thick,samples=200][domain=0:0.16667] {-1/2*x+1/4+1/6*x-1/36};
\addplot[blue, thick,samples=200][domain=0.16667:0.25] {-1/2*x+1/4};
\addplot[blue, thick,samples=200][domain=0.25:0.75] {1/8};
\addplot[blue, thick,samples=200][domain=0.75:0.83333] {1/2*x-1/4};
\addplot[blue, thick,samples=200][domain=0.83333:1] {1/2*x-1/4-1/6*x+5/36};
\addplot[olive, dashed, very thick,samples=200][domain=0:0.33333] {2/9-7/24*x};
\addplot[olive, dashed, very thick,samples=200][domain=0.33333:0.66667] {1/8};
\addplot[olive, dashed, very thick,samples=200][domain=0.66667:1] {2/9+7/24*(x-1)};
\node[fill, olive, circle, inner sep=1.5pt] 	at (axis cs:0.5,0.125){};
\node[fill, olive, circle, inner sep=1.5pt] 	at (axis cs:0,0.22222){};
\node[fill, olive, circle, inner sep=1.5pt] 	at (axis cs:1,0.22222){};
\end{axis}
\end{tikzpicture}
}

\caption{Optimal information with concavified costs}
\end{figure}
\end{center}

It should be noted that the non-concavity of $c$ can only be inferred by
counterfactuals: there is no way for the DM to exploit the lack of 
concavity. This failure of concavity can only appear in regions where
the mean-preserving contraction constraint (\ref{MPC}) binds,
as $\sigma_A$ is optimal for each $A$: otherwise, there would be a 
mean-preserving spread that would be feasible and reduces information 
costs, while simultaneous increasing indirect utility (since $\phi_A$ 
is convex). Thus the non-concavity can only be inferred indirectly by 
the acyclicality condition for the menus that appear in the dataset.

That being said, the case where $c$ is not concave lends itself to a 
certain lack of robustness. For instance, if there were to be additional 
data from some menu $B$, it might not be possible for there to be 
\emph{any} dataset $\sigma_B$ which satisfies our axioms. This is because 
it might be that, the revealed posterior mean $z_{\sigma_B}(b)$ for any 
optimal distribution $\sigma_B$ lies in a non-concave region of $c$, i.e. 
there exist $z_1,z_2$ and $\alpha\in(0,1)$ such that 
$\alpha c(z_1)+(1-\alpha)c(z_2)>c(z_{\sigma_B}(b))$
and $I_{F_0,F_{\sigma_B}}(z)>0,\forall z\in [z_1,z_2]$. In that case, we 
would get a violation of Axiom \ref{NIPMC}, as according to NIPMC, if we
could reallocate posterior means $z_1,z_2$ from another decision problem
with menu $A$, we could improve the DM's payoff for large enough 
coefficients $\beta_B$. If so, we would not be able to identify 
such non-concave $c$ using our Theorem \ref{costexist}.

To see this in the context of Example 4 above, consider the singleton 
menu $B=\{b\}$. Then the revealed posterior is $z=\frac{1}{2}$. Of course, 
however, the optimal information acquisition would involve posterior means 
at $z=\frac{1}{6}$ and $\frac{5}{6}$, so the revealed posterior is not
actually optimal with respect to $c$.

Nevertheless, it is possible to check sufficient conditions for the
concavity of $c$ by use of a linear programming algorithm. We spell out
the algorithm in the appendix, which involves checking $n^{\vert Z\vert}$
sets of $n\vert Z\vert$ inequalities. 

In any case, the reader is reminded that the standard case where $c$ is 
concave, and hence canonical, is also covered by our representation theorem.
That being said, we argue that $c$ being non-concave can actually be reasonable 
in some circumstances. This is because there may be instances in which one 
would want to capture the idea that distinguishing between two lower states 
with lower values of $z$ may be costly in ways that distinguishing from those 
with higher values may not. Thus, in Example 4, it is quite costly to 
distinguish between $z=0$ and $z=\frac{1}{3}$, leading the DM to optimally
pool them to reduce costs. Our cost function captures this, while
maintaining the posterior-mean separable representation, and hence enabling
the use of the tools from \cite{dworczak2019simple}.

\section{Conclusion}

Our main results have established that the optimization with respect to
PMS costs is behaviorally equivalent to the conditions NIAS and NIPMC. We 
argued that the latter is capturing the idea of optimizing via a Dworczak-%
Martini ``price function" through the relevant acyclicality condition. The 
presence of such cost functions enables the use of these price functions
for the optimal information choice problem. We hope that, with further 
understanding of the implications of this class of information costs, they 
will be used further in applied theoretical work.

Several open questions remain. First, while we have presented an example
of a natural case where PMS costs apply (reduction of posterior variance),
it would be interesting to find other such cases. One such possibility
is via a Brownian diffusion process over posterior means, where flow costs 
only depend on the current posterior mean. By applying the arguments in 
\cite{morris2019wald}, this leads to a representation of costs over $[0,1]$ 
that is posterior-mean separable. By the martingale representation 
theorem, it is possible to generate any mean-preserving spread of the prior
mean via some stopping rule. However, the stopping rule needed to do so 
must look at the entire distribution in the filtration at time $t$, since
it needs to ensure that it remains a mean-preserving contraction of the 
prior $F_0$. Providing a natural environment for this sort of Brownian
diffusion process, or a modification thereof, would extend the circumstances
under which PMS costs would be appropriate.

In addition, we saw that the cost function $c$ that is consistent with the
data need not be concave. As discussed, this means that this representation 
of the cost function would not be canonical, and so, counterfactually,
there may be circumstances in which increased information would yield 
lower costs. While we have given sufficient conditions to check whether
there is a concave $c$ that rationalizes the dataset, it would be useful
to provide a tighter characterization of necessary and sufficient conditions
that guarantee concavity.

\bibliographystyle{jpe.bst}
\bibliography{PosteriorMean}

\newcommand{\noop}[1]{}
\begin{thebibliography}{41}
\newcommand{\enquote}[1]{``#1''}
\providecommand{\natexlab}[1]{#1}
\providecommand{\url}[1]{\texttt{#1}}
\providecommand{\urlprefix}{URL }

\bibitem[{Afriat(1967)}]{afriat1967construction}
Afriat, Sydney~N. 1967.
\newblock \enquote{The construction of utility functions from expenditure data.}
\newblock \emph{International economic review} 8~(1):67--77.

\bibitem[{Aliprantis and Border(2006)}]{aliprantis2006infinite}
Aliprantis, Charalambos~D and Kim Border. 2006.
\newblock \emph{Infinite Dimensional Analysis: A Hitchhiker's Guide}.
\newblock Springer Science \& Business Media.

\bibitem[{Arieli et~al.(2020)Arieli, Babichenko, Smorodinsky, and Yamashita}]{arieli2020optimal}
Arieli, Itai, Yakov Babichenko, Rann Smorodinsky, and Takuro Yamashita. 2020.
\newblock \enquote{Optimal Persuasion via Bi-Pooling.} Working Paper.

\bibitem[{Aumann and Maschler(1995)}]{Aumann1995}
Aumann, Robert~J and Michael Maschler. 1995.
\newblock \emph{Repeated games with incomplete information}.
\newblock MIT press.

\bibitem[{Bloedel and Zhong(2020)}]{bloedel2020cost}
Bloedel, Alexander~W and Weijie Zhong. 2020.
\newblock \enquote{The cost of optimally-acquired information.} Working Paper.

\bibitem[{Caplin and Dean(2013)}]{caplin2013behavioral}
Caplin, Andrew and Mark Dean. 2013.
\newblock \enquote{Behavioral implications of rational inattention with shannon entropy.}
\newblock Tech. rep., National Bureau of Economic Research.

\bibitem[{Caplin and Dean(2015)}]{caplin2015revealed}
---{}---{}---. 2015.
\newblock \enquote{Revealed preference, rational inattention, and costly information acquisition.}
\newblock \emph{American Economic Review} 105~(7):2183--2203.

\bibitem[{Caplin, Dean, and Leahy(2022)}]{caplin2022rationally}
Caplin, Andrew, Mark Dean, and John Leahy. 2022.
\newblock \enquote{Rationally inattentive behavior: Characterizing and generalizing Shannon entropy.}
\newblock \emph{Journal of Political Economy} 130~(6):1676--1715.

\bibitem[{Caplin and Martin(2015)}]{caplin2015testable}
Caplin, Andrew and Daniel Martin. 2015.
\newblock \enquote{A testable theory of imperfect perception.}
\newblock \emph{The Economic Journal} 125~(582):184--202.

\bibitem[{Chambers, Liu, and Rehbeck(2020)}]{chambers2020costly}
Chambers, Christopher~P, Ce~Liu, and John Rehbeck. 2020.
\newblock \enquote{Costly information acquisition.}
\newblock \emph{Journal of Economic Theory} 186:104979.

\bibitem[{de~Oliveira(2014)}]{de2014axiomatic}
de~Oliveira, Henrique. 2014.
\newblock \enquote{Axiomatic foundations for entropic costs of attention.}
\newblock Tech. rep., Technical report, Mimeo.

\bibitem[{De~Oliveira et~al.(2017)De~Oliveira, Denti, Mihm, and Ozbek}]{de2017rationally}
De~Oliveira, Henrique, Tommaso Denti, Maximilian Mihm, and Kemal Ozbek. 2017.
\newblock \enquote{Rationally inattentive preferences and hidden information costs.}
\newblock \emph{Theoretical Economics} 12~(2):621--654.

\bibitem[{Denti(2022)}]{denti2022posterior}
Denti, Tommaso. 2022.
\newblock \enquote{Posterior separable cost of information.}
\newblock \emph{American Economic Review} 112~(10):3215--59.

\bibitem[{Denti, Marinacci, and Rustichini(2022)}]{denti2022experimental}
Denti, Tommaso, Massimo Marinacci, and Aldo Rustichini. 2022.
\newblock \enquote{Experimental cost of information.}
\newblock \emph{American Economic Review} 112~(9):3106--23.

\bibitem[{Dizdar and Kov{\'a}{\v{c}}(2020)}]{dizdar2020simple}
Dizdar, Deniz and Eugen Kov{\'a}{\v{c}}. 2020.
\newblock \enquote{A Simple Proof of Strong Duality in the Linear Persuasion Pccroblem.}
\newblock \emph{Games and Economic Behavior} 122:407--412.

\bibitem[{Dworczak and Kolotilin(2019)}]{dworczak2019persuasion}
Dworczak, Piotr and Anton Kolotilin. 2019.
\newblock \enquote{The persuasion duality.} Working Paper.

\bibitem[{Dworczak and Martini(2019)}]{dworczak2019simple}
Dworczak, Piotr and Giorgio Martini. 2019.
\newblock \enquote{{The Simple Economics of Optimal Persuasion}.}
\newblock \emph{Journal of Political Economy} 127~(5):1993--2048.

\bibitem[{Ellis(2018)}]{ellis2018foundations}
Ellis, Andrew. 2018.
\newblock \enquote{Foundations for optimal inattention.}
\newblock \emph{Journal of Economic Theory} 173:56--94.

\bibitem[{Gentzkow and Kamenica(2016)}]{gentzkow2016rothschild}
Gentzkow, Matthew and Emir Kamenica. 2016.
\newblock \enquote{{A Rothschild-Stiglitz Approach to Bayesian Persuasion}.}
\newblock \emph{American Economic Review} 106~(5):597--601.

\bibitem[{H{\'e}bert and Woodford(2023)}]{hebert2023rational}
H{\'e}bert, Benjamin and Michael Woodford. 2023.
\newblock \enquote{Rational inattention when decisions take time.}
\newblock \emph{Journal of Economic Theory} :105612.

\bibitem[{Kamenica and Gentzkow(2011)}]{Kamenica2011}
Kamenica, Emir and Matthew Gentzkow. 2011.
\newblock \enquote{{Bayesian Persuasion}.}
\newblock \emph{American Economic Review} 101~(October):2590--2615.

\bibitem[{Kleiner, Moldovanu, and Strack(2021)}]{kleiner2021extreme}
Kleiner, Andreas, Benny Moldovanu, and Philipp Strack. 2021.
\newblock \enquote{Extreme points and majorization: Economic applications.}
\newblock \emph{Econometrica} 89~(4):1557--1593.

\bibitem[{Kolotilin et~al.(2017)Kolotilin, Mylovanov, Zapechelnyuk, and Li}]{kolotilin2017persuasion}
Kolotilin, Anton, Tymofiy Mylovanov, Andriy Zapechelnyuk, and Ming Li. 2017.
\newblock \enquote{Persuasion of a privately informed receiver.}
\newblock \emph{Econometrica} 85~(6):1949--1964.

\bibitem[{Kolotilin and Zapechelnyuk(2019)}]{kolotilin2019persuasion}
Kolotilin, Anton and Andriy Zapechelnyuk. 2019.
\newblock \enquote{Persuasion meets delegation.} Working Paper.

\bibitem[{Kreutzkamp(2022)}]{kreutzkamp2022endogenous}
Kreutzkamp, Sophie. 2022.
\newblock \enquote{Endogenous Information Acquisition in Cheap-Talk Games.} Working Paper.

\bibitem[{Lin(2022)}]{lin2022stochastic}
Lin, Yi-Hsuan. 2022.
\newblock \enquote{Stochastic choice and rational inattention.}
\newblock \emph{Journal of Economic Theory} 202:105450.

\bibitem[{Lipnowski and Ravid(2022)}]{lipnowski2022predicting}
Lipnowski, Elliot and Doron Ravid. 2022.
\newblock \enquote{Predicting Choice from Information Costs.} Working Paper.

\bibitem[{Mat{\v{e}}jka and McKay(2015)}]{matvejka2015rational}
Mat{\v{e}}jka, Filip and Alisdair McKay. 2015.
\newblock \enquote{Rational inattention to discrete choices: A new foundation for the multinomial logit model.}
\newblock \emph{American Economic Review} 105~(1):272--98.

\bibitem[{Mensch(2018)}]{mensch2018cardinal}
Mensch, Jeffrey. 2018.
\newblock \enquote{Cardinal representations of information.} Working Paper.

\bibitem[{Mensch(2021)}]{mensch2021rational}
---{}---{}---. 2021.
\newblock \enquote{Rational inattention and the monotone likelihood ratio property.}
\newblock \emph{Journal of Economic Theory} 196:105284.

\bibitem[{Mensch and Ravid(2024)}]{mensch2024monopoly}
Mensch, Jeffrey and Doron Ravid. 2024.
\newblock \enquote{Monopoly, Product Quality, and Flexible Learning.} Working Paper.

\bibitem[{Morris and Strack(2019)}]{morris2019wald}
Morris, Stephen and Philipp Strack. 2019.
\newblock \enquote{The wald problem and the relation of sequential sampling and ex-ante information costs.} Working Paper.

\bibitem[{Mussa and Rosen(1978)}]{mussa1978monopoly}
Mussa, Michael and Sherwin Rosen. 1978.
\newblock \enquote{Monopoly and product quality.}
\newblock \emph{Journal of Economic theory} 18~(2):301--317.

\bibitem[{Pomatto, Strack, and Tamuz(2023)}]{pomatto2023cost}
Pomatto, Luciano, Philipp Strack, and Omer Tamuz. 2023.
\newblock \enquote{The Cost of Information.}
\newblock \emph{American Economic Review} forthcoming.

\bibitem[{Ravid, Roesler, and Szentes(2022)}]{ravid2022learning}
Ravid, Doron, Anne-Katrin Roesler, and Bal{\'a}zs Szentes. 2022.
\newblock \enquote{Learning before trading: on the inefficiency of ignoring free information.}
\newblock \emph{Journal of Political Economy} 130~(2):346--387.

\bibitem[{Rochet(1987)}]{rochet1987necessary}
Rochet, Jean-Charles. 1987.
\newblock \enquote{A necessary and sufficient condition for rationalizability in a quasi-linear context.}
\newblock \emph{Journal of Mathematical Economics} 16~(2):191--200.

\bibitem[{Rockafellar(1966)}]{rockafellar1966characterization}
Rockafellar, Ralph. 1966.
\newblock \enquote{Characterization of the subdifferentials of convex functions.}
\newblock \emph{Pacific Journal of Mathematics} 17~(3):497--510.

\bibitem[{Rothschild and Stiglitz(1970)}]{rothschild1970increasing}
Rothschild, Michael and Joseph~E Stiglitz. 1970.
\newblock \enquote{{Increasing Risk: I. A Definition}.}
\newblock \emph{Journal of Economic Theory} 2~(3):225--243.

\bibitem[{Thereze(2022)}]{thereze2022adverse}
Thereze, Jo\~{a}o. 2022.
\newblock \enquote{Adverse Selection and Endogenous Information.} Working Paper.

\bibitem[{Whitmeyer and Zhang(2022)}]{whitmeyer2022costly}
Whitmeyer, Mark and Kun Zhang. 2022.
\newblock \enquote{Costly Evidence and Discretionary Disclosure.} Working paper.

\bibitem[{Zapechelnyuk(2020)}]{zapechelnyuk2020optimal}
Zapechelnyuk, Andriy. 2020.
\newblock \enquote{Optimal quality certification.}
\newblock \emph{American Economic Review: Insights} 2~(2):161--76.

\end{thebibliography}

\appendix

\section{Proofs of preliminary results}

In several of the proofs, we make reference to the main results of 
\cite{dworczak2019simple}, listed below as Lemma \ref{DMT12}.

\begin{lemma}[\cite{dworczak2019simple}, Theorems 1 and 2]
    \label{DMT12}
    Given Lipschitz-continuous payoff function 
    $V: [0,1]\rightarrow\mathbb{R}$, the distribution 
    $F^*\in \mathcal{I}_{F_0}$ is a solution to 
    \[
    \max_{F\in \mathcal{I}_{F_0}} \int V(z)dF(z)
    \]
    if and only if there exists convex function 
    $P:[0,1]\rightarrow \mathbb{R}$ such that 
    \begin{equation}
    \label{major}
        P(z)\geq V(z),\forall z\in [0,1],
    \end{equation}
    \begin{equation}
        \label{samesupp}
        \mbox{supp}(F^*)\subset \{z\in[0,1]:V(z)=P(z)\}
    \end{equation}
    \begin{equation}
        \label{sameint}
        \int_0^1 P(z)dF^*(z)=\int_0^1 P(z)dF_0(z)
    \end{equation}
\end{lemma}

We now proceed with the proofs of the results that appear in the body
of the paper.

\begin{proof}[Proof of Lemma \ref{DM2019}]

    Given prior $F_0$, the DM's problem is given by 
    \begin{equation}
    \label{DMprob}
        \max_{F\in\mathcal{I}_{F_0}} \int [\phi_A (z)+c(z)]dF(z)
    \end{equation}
    Suppose that $F_{\sigma_A}$ is a solution to \ref{DMprob}. 
    By Lemma \ref{DMT12}, there exists function
    $P:[0,1]\rightarrow\mathbb{R}$ that satisfies (\ref{major})-%
    (\ref{sameint}) with respect to $V(z)=\phi_A(z)+c(z)$. By Lemma 2 of 
    \cite{dworczak2019simple}, $P$ must be affine on any interval where 
    $I_{F_0,F_{\sigma_A}}(z)>0$.
        
    It thus remains to be shown that the converse holds as well, i.e. 
    that for $\hat{F}_{\sigma_A}\in\mathcal{I}_{\hat{F}_0}$, if one takes 
    the same $P$ as before, $I_{F_0,F_{\sigma_{A}}}(z)=0\implies I_{\hat{F}_0,\hat{F}_{\sigma_A}}(z)=0$,
    and (\ref{samesupp}) holds, then $P$ is still affine wherever 
    $I_{\hat{F}_0,\hat{F}_{\sigma_A}}(z)>0$, and (\ref{sameint}) is 
    satisfied. To see this, since $I_{F_0,F_{\sigma_{A}}}(z)=0\implies I_{\hat{F}_0,\hat{F}_{\sigma_A}}(z)=0$, 
    we have that on any interval $[z_1,z_2]$ such that 
    $I_{\hat{F}_0,\hat{F}_{\sigma_A}}(z)>0$, $P$ must
    be affine since for all $z\in(z_1,z_2)$, $I_{F_0,F_{\sigma_A}}(z)>0$. To show that (\ref{sameint}) holds, for each $z\in Z$ such that 
    $I_{\hat{F}_0,\hat{F}_{\sigma_A}}(z)>0$, define the probability 
    weights $\gamma_{-}^z,\gamma_=^z,\gamma_+^z$ as the conditional 
    probabilities that, if the true state is $z$, the posterior belief
    will be $\hat{z}<z$, $\hat{z}=z$, and $\hat{z}>z$, respectively. 
    Next, note that for any maximal interval $(z_1,z_2)$ such that $I_{\hat{F}_0,\hat{F}_{\sigma_A}}(z)>0$,
    one has that 
    \begin{equation}
    \label{equaloninterval}
    \int_{(z_1,z_2)}d\hat{F}_{\sigma_A}(z)=\int_{(z_1,z_2)}d\hat{F}_0(z)+
    \gamma_+^{z_1}\hat{f}_0(z_1)+\gamma_-^{z_2}\hat{f}_0(z_2)
    \end{equation}
    where $\hat{f}_0(z)$ is the prior probability of state $z\in Z$.
    Furthermore, the function $P$ must be affine on $[z_1,z_2]$. So, 
    letting $\cup (z_1,z_2)$ be the union of such intervals, we have
    \begin{align*}
        \int_0^1 [\phi_A(z)+c(z)]d\hat{F}_{\sigma_A}(z) & =\int_0^1 P(z)d\hat{F}_{\sigma_A}(z)\\
        & =\int_{[0,1]\setminus \{\cup (z_1,z_2)\}} P(z)d\hat{F}_{\sigma_A}(z)+\sum_{(z_1,z_2)}\int_{(z_1,z_2)}P(z)d\hat{F}_{\sigma_A}(z)\\
        & =\int_{[0,1]\setminus \{\cup (z_1,z_2)\}} \gamma_=^z P(z)d\hat{F}_0(z)+\sum_{(z_1,z_2)}[\int_{(z_1,z_2)}P(z)d\hat{F}_0(z)\\
        & \qquad +\gamma_+^{z_1}P(z_1)\hat{f}_0(z_1)+\gamma_-^{z_2}P(z_2)\hat{f}_0(z_2)]\\
        & =\int_{[0,1]\setminus \{\cup (z_1,z_2)\}} (\gamma_-^z+\gamma_=^z+\gamma_+^z)P(z)d\hat{F}_0(z)+\sum_{(z_1,z_2)}\int_{(z_1,z_2)}P(z)d\hat{F}_0(z)\\
        &=\int_0^1 P(z)d\hat{F}_0(z)
    \end{align*}
    where the first equality is due to (\ref{samesupp}), the third 
    equality is due to the affineness of $P$ on $[z_1,z_2]$, and the last
    equality is due to the fact that $\gamma_{-}^z+\gamma_=^z+\gamma_+^z=1$. 
    Thus (\ref{sameint}) holds as required.
\end{proof}

\section{Proof of Theorem \ref{costexist}}

To prove the result, we require some additional lemmas.
\begin{lemma}
    \label{finiteineq}
    For $i=1,...,n$, define priors $F_0^i\in\mathcal{F}$ such that 
    $\mbox{supp}(F_0^i)\subset Z$ and distributions of revealed posterior 
    means $F_i\in \mathcal{I}_{F_0^i}$ such that:
    \begin{enumerate}[(i)]
    \item $\mbox{supp}(F_i)\subset\mbox{supp}(F_{\sigma_{A_i}})$, and
    \item For all $z\in\mbox{supp}(F_0)$,
    \begin{equation}
    \label{sameinfoconstraint}
    I_{F_0,F_{\sigma_{A_i}}}(z)=0\implies I_{F^i_0,F_i}(z)=0
    \end{equation}
    \end{enumerate}
    Then for all $z\in [0,1]$, (\ref{sameinfoconstraint}) holds.
\end{lemma}

Lemma \ref{finiteineq} shows that in order to demonstrate that \ref{allineq}
holds, it is sufficient to look at a finite number of such inequalities, namely 
those at $z\in Z$.

\begin{proof}[Proof of Lemma \ref{finiteineq}]
    For all $z\in [0,1]$ such that $I_{F_0,F_{\sigma_{A_i}}}(z)=0$ and 
    $z\notin Z$, let $z_1 =\max\{\hat{z}\in Z, \hat{z}<z\}$ 
    and $z_2 =\min\{\hat{z}\in Z, \hat{z}>z\}$. Then for all 
    $\hat{z}\in (z_1,z_2)$, $\hat{z}\notin\mbox{supp}(F_{\sigma_{A_i}})$, or 
    else $I_{F_0,F_{\sigma_{A_i}}}(z)>0$. So, $\hat{z}\notin \mbox{supp}(F_i)$ 
    as well. Thus if 
    \[
    I_{F^i_0,F_i}(z_1)=0=I_{F^i_0,F_i}(z_2)
    \]
    we conclude that for all $z\in (z_1,z_2)$,
    \[
    I_{F_0,F_{\sigma_{A_i}}}(z)=0=I_{F^i_0,F_i}(z).
    \]
\end{proof}

\begin{lemma}
\label{exanteNIPMC}
    Let $A_1,...,A_n$ be a sequence of menus in $\mathcal{A}$, and 
    $\beta_{(i,j)}:A_i\times A_j\rightarrow \mathbb{R}_+^2$ for $i,j=1,...n$. 
    Suppose that Axioms \ref{NIAS} and \ref{NIPMC} hold, and the following conditions hold:
    \begin{enumerate}
        \item 
    For all $i$ and 
    $z^*\in \{z: I_{F_0,F_{\sigma_{A_i}}}(z)=0\}\cap Z$,
    \begin{multline}
    \label{reallocMPCnobind}
    \sum_{j=1}^n \int_0^{z^*}\sum_{s\in Z} [\sum_{a_j\in A_j}\sum_{\{a_i: z_{\sigma_{A_i}}(a_i)\leq z\}}\beta_{(i,j)}(a_i,a_j)\sigma_{A_i}(a_i\vert s)\\
    -\sum_{a_i\in A_i}\sum_{\{a_j: z_{\sigma_{A_j}}(a_j)\leq z\}}\beta_{(j,i)}(a_j,a_i)\sigma_{A_j}(a_j\vert s)]f_0(s)dz\geq 0
    \end{multline}
    with equality at $z^*=1$.

    \item For all $i$, 
       \begin{equation}
        \label{reallocsameprior}
        \sum_{j=1}^n \sum_{z \in Z}\sum_{a_i \in A_i} \sum_{a_j \in A_j } \beta_{(i,j)}(a_i, a_j) \sigma_{A_i}(a_i|z)f_0(z)=\sum_{j=1}^n \sum_{z \in Z}\sum_{a_i \in A_i} \sum_{a_j \in A_j } \beta_{(j,i)}(a_j, a_j) \sigma_{A_j}(a_j|z)f_0(z).
    \end{equation}
    \end{enumerate}
    
    Then 
    \begin{equation}
        \label{exanteNIPMCeq}
        \sum_{i=1}^n\sum_{j=1}^n\sum_{z\in Z}\sum_{a_j\in A_j}\sum_{a_i\in A_i}\beta_{(i,j)}(a_i,a_j) [u(a_i,z)-u(a_j,z)]\sigma_{A_i}(a_i\vert z)f_0(z)\geq0
    \end{equation}
\end{lemma}

\begin{proof}[Proof of Lemma \ref{exanteNIPMC}]
    For each $i$, let
    \[
    \gamma_i\coloneqq\sum_{j=1}^n \sum_{z\in Z}\sum _{a_i\in A_i}\sum_{a_j\in A_j}\beta_{(i,j)}(a_i,a_j)\sigma_{A_i}(a_i\vert z)f_0(z)
    \]

     By equation (\ref{reallocsameprior}), we also have
    \[
    \gamma_i=\sum_{j=1}^n \sum_{z\in Z}\sum_{a_j\in A_j}\sum_{a_i\in A_i}\beta_{(j,i)}(a_j,a_i)\sigma_{A_j}(a_j\vert z)f_0(z)
    \]
    If $\gamma_i>0$, define $F_i,G_i\in \mathcal{F}$ by 
    \[
    F_i(z)=\frac{1}{\gamma_i}\sum_{j=1}^n \sum_{s\in Z}\sum_{a_j\in A_j} \sum_{\{a_i: z_{\sigma_{A_i}}(a_i)\leq z\}}\beta_{(i,j)}(a_i,a_j)\sigma_{A_i}(a_i\vert s)f_0(s)
    \]
    \[
    G_i(z)=\frac{1}{\gamma_i}\sum_{j=1}^n \sum_{s\in Z} \sum_{a_i\in A_i}\sum_{\{a_j: z_{\sigma_{A_j}}(a_j)\leq z\}}\beta_{(j,i)}(a_j,a_i)\sigma_{A_j}(a_j\vert s)f_0(s)
    \]
    If $\gamma_i=0$, set $F_i=F_0=G_i$. Notice that both $F_i$ and $G_i$
    are increasing functions of $z$, $F_i(0)=G_i(0)=0$, and, by  (\ref{reallocsameprior}), $F_i(1)=G_i(1)=1$; so, both $F_i$ and $G_i$
    are valid CDFs.

    By construction, the support of $F_i$ is the set 
    $\{z:\exists a_i\in A_i: z=z_{\sigma_{A_i}}(a_i)\}$, which is just 
    $\mbox{supp}(F_{\sigma_{A_i}})$. 

    Now let $F_0^i(z)$ have support on $\{z: I_{F_0,F_{\sigma_{A_i}}}(z)=0\}\cap Z$ such that, for $z^*\in \mbox{supp}(F_0^i)$, 
    $\int_0^{z^*} F_0^i(s)ds=\int_0^{z^*} F_i(s)ds$. 
    Then $F_i,G_i\in \mathcal{I}_{F_0^i}$, seen as follows. Notice that for 
    any two consecutive $z_1,z_2\in \mbox{supp}(F_0^i)$, $\int_0^z F_0^i(s)ds$ 
    is affine for $z\in[z_1,z_2]$; however, on the same interval, both 
    $\int_0^z F_i(s)ds$ and $\int_0^z G_i(s)ds$ are convex. Moreover, by 
    construction, at $z^*\in\mbox{supp}(F_0^i)$, 
    $\int_0^{z^*}F_i(s)ds\geq \int_0^{z^*}G_i(s)ds$ by (\ref{reallocMPCnobind}). So, for all 
    $z\in[0,1]$, 
    \[
    \int_0^z F_0^i(s)ds\geq \max\{\int_0^z F_i(s)ds, \int_0^z G_i(s)ds\}
    \]
    and so $F_0^i$ is a mean-preserving spread of both $F_i$ and $G_i$, with
    $I_{F_0^i,F_i}(z^*)=0$ for $z^*\in \mbox{supp}(F_0^i)$.

    Lastly, summing up over $i=1,...,n$,
    \begin{align*}
        \sum_{i=1}^n \gamma_i F_i(z) & =\sum_{i=1}^n \sum_{j=1}^n \sum_{s\in Z}\sum_{a_j\in A_j} \sum_{\{a_i: z_{\sigma_{A_i}}(a_i)\leq z\}}\beta_{(i,j)}(a_i,a_j)\sigma_{A_i}(a_i\vert s)f_0(s)\\
        & =\sum_{j=1}^n\sum_{i=1}^n \sum_{s\in Z}\sum_{a_i\in A_i} \sum_{\{a_j: z_{\sigma_{A_j}}(a_j)\leq z\}}\beta_{(j,i)}(a_j,a_i)\sigma_{A_j}(a_j\vert s)f_0(s)\\
        & = \sum_{i=1}^n \gamma_i G_i(z)
    \end{align*}

    Therefore, $\{G_i\}_{i=1}^n$ is a reallocation of posterior means from 
    $\{F_i\}_{i=1}^n$ given priors $\{F_0^i\}_{i=1}^n$. By Axiom \ref{NIPMC},
    \begin{equation}
    \label{L5NIPMC}
    \sum_{i=1}^n \gamma_i\int_0^1 \phi_{A_i}(z)dF_i(z)\geq \sum_{i=1}^n \gamma_i\int_0^1 \phi_{A_i}(z)dG_i(z)
    \end{equation}
    Substitution yields, by (\ref{Bayes}) and Axiom \ref{NIAS},
    \begin{align*}
        \sum_{i=1}^n \gamma_i\int_0^1 \phi_{A_i}(z)dF_i(z) &=\sum_{i=1}^n \sum_{j=1}^n\sum_{\hat{z}\in Z}\sum_{z\in\mbox{supp}(F_i)}\sum_{a_j\in A_j} \sum_{a_i\in A_i}\beta_{(i,j)}(a_i,a_j)\phi_{A_i}(z)\sigma_{A_i}(a_i\vert \hat{z})f_0(\hat{z}) \\
        &=\sum_{i=1}^n \sum_{j=1}^n\sum_{z\in\mbox{supp}(F_i)}\sum_{a_j\in A_j} \sum_{a_i\in A_i}\beta_{(i,j)}(a_i,a_j)u(a_i,z)D_{A_i}(a_i\vert z)f_i(z)\\
        &= \sum_{i=1}^n \sum_{j=1}^n\sum_{z\in Z}\sum_{a_j\in A_j} \sum_{a_i\in A_i} \beta_{(i,j)}(a_i,a_j)u(a_i,z)\sigma_{A_i}(a_i\vert z)f_0(z) \numberthis \label{EUFI}
    \end{align*}
    
    Similarly, 
    \begin{align*}
        \sum_{i=1}^n \gamma_i\int_0^1 \phi_{A_i}(z)dG_i(z) &=\sum_{i=1}^n \sum_{j=1}^n\sum_{\hat{z}\in Z}\sum_{z\in\mbox{supp}(G_i)}\sum_{a_i\in A_i} \sum_{a_j\in A_j}\beta_{(j,i)}(a_j,a_i)\phi_{A_i}(z)\sigma_{A_j}(a_j\vert \hat{z})f_0(\hat{z}) \\
        &=\sum_{i=1}^n \sum_{j=1}^n\sum_{z\in\mbox{supp}(G_i)}\sum_{a_i\in A_i} \sum_{a_j\in A_j}\beta_{(j,i)}(a_j,a_i)\{\max_{\hat{a}_i\in A_i} u(\hat{a}_i,z)\}\sigma_{A_j}(a_j) \\
        &\geq \sum_{i=1}^n \sum_{j=1}^n\sum_{z\in\mbox{supp}(G_i)}\sum_{a_i\in A_i} \sum_{a_j\in A_j}\beta_{(j,i)}(a_j,a_i)u(a_i,z)\sigma_{A_j}(a_j) \\
        &= \sum_{i=1}^n \sum_{j=1}^n\sum_{z\in Z}\sum_{a_i\in A_i} \sum_{a_j\in A_j}\beta_{(j,i)}(a_j,a_i)u(a_i,z)\sigma_{A_j}(a_j\vert z)f_0(z) \\
        &=  \sum_{i=1}^n \sum_{j=1}^n\sum_{z\in Z} \sum_{a_i\in A_i}\sum_{a_j\in A_j}\beta_{(i,j)}(a_i,a_j)u(a_j,z)\sigma_{A_i}(a_i\vert z)f_0(z)  \numberthis \label{EUGI}
    \end{align*}
    where the last equality is from switching the indices $i,j$. Combining
    (\ref{L5NIPMC})-(\ref{EUGI}) yields the desired result.
\end{proof}

\begin{lemma}
    \label{farkas}
    Let $\mathbf{A}\in\mathbb{R}^{m\times (n+k)}$ and 
    $\mathbf{b}\in \mathbb{R}^{m}$. Exactly one of the following holds:
    \begin{enumerate}
    \item There exists $\lambda \in \mathbb{R}^{n+k}$ such that 
    $\mathbf{A}\lambda \leq \mathbf{b}$ and, for each $i=1,...,n+k$, 
    \begin{align*}
        \lambda_i &\in \mathbb{R}, \qquad i=1,...,n\\
        &\geq 0, \qquad i=n+1,...,n+k.
    \end{align*}
    \item There exists $\mathbf{y}\in \mathbb{R}^{m}$ such that:
    \begin{enumerate}[(i)]
        \item $\mathbf{y}\geq 0$, 
        \item For $i=1,...,n+k$,
        \begin{align*}
        (\mathbf{A}^T \mathbf{y})_i& =0,\qquad i=1,...,n\\
        &\geq 0,\qquad i=n+1,...,n+k
        \end{align*}
        \item $\mathbf{b}\cdot \mathbf{y}<0$.
    \end{enumerate}
    \end{enumerate}
\end{lemma}

\begin{proof}[Proof of Lemma \ref{farkas}]
    Let $\mathbf{A}^\prime \in \mathbb{R}^{m^\prime\times n^\prime}$ and 
    $\mathbf{b}^\prime \in \mathbb{R}^{m^\prime}$. By Farkas' lemma 
    (\cite{aliprantis2006infinite}, Corollary 5.85), exactly one of the 
    following holds:
    \begin{enumerate}[(i)]
        \item There exists $\mathbf{x}^\prime\in \mathbb{R}^{n^\prime}_+$ such 
        that $\mathbf{A}^\prime \mathbf{x}^\prime = \mathbf{b}^\prime$.
        \item There exists $\mathbf{y}^\prime \in \mathbb{R}^{m^\prime}$
        such that $\mathbf{A^\prime}^T \mathbf{y}^\prime\geq 0$ and $\mathbf{b}^\prime\cdot \mathbf{y}^\prime<0$. 
    \end{enumerate}

    Let $\mathbf{A}_{1,...n}$ be the $m\times n$ matrix consisting of the first 
    $n$ columns of $\mathbf{A}$, and $\mathbf{A}_{n+1,...,n+k}$ be the 
    $m\times k$ matrix consisting of the last $k$ columns. Condition (2)(i)-(ii)
    in the lemma is equivalent to $y$ satisfying
    \begin{equation}
    \label{farkas1}
        \begin{pmatrix} \mathbf{A}_{1,...n}^T \\ -\mathbf{A}_{1,...n}^T \\ \mathbf{A}_{n+1,...,n+k}^T \\ \mathbf{I} \end{pmatrix}\mathbf{y}\geq 0
    \end{equation}
    where $\mathbf{I}$ is the $m\times m$ identity matrix. By Farkas' lemma,
    (\ref{farkas1}) holds for some $\mathbf{y}\in \mathbb{R}^m$ and 
    $\mathbf{b}\cdot \mathbf{y}\leq 0$ 
    if and only if
    \begin{equation}
    \label{farkas2}
        \{\mathbf{x}: \begin{pmatrix} \mathbf{A}_{1,...n} & -\mathbf{A}_{1,...n} & \mathbf{A}_{n+1,...,n+k} & I \end{pmatrix} \mathbf{x} = \mathbf{b}, \mathbf{x} \geq 0\}=\emptyset 
    \end{equation}
    In turn, (\ref{farkas2}) holds if and only if there is no $\lambda$
    such that 
    \begin{align*}
    \label{lambda}
        \lambda_i & \in \mathbb{R}, \qquad i=1,...,n\\
        & \geq 0, \qquad i=n+1,...,n+k \numberthis
    \end{align*}
    and $\mathbf{s}\geq 0$ such that $\mathbf{A}\lambda+\mathbf{s}=\mathbf{b}$.
    This, of course, is equivalent to there being no $\lambda$ as defined in 
    (\ref{lambda}) such that $\mathbf{A}\lambda \leq \mathbf{b}$.
\end{proof}

\begin{proof}[Proof of Theorem \ref{costexist}]
We start with the necessity of NIAS and NIPMC. Let 
$c:[0,1]\rightarrow \mathbb{R}$ be Lipschitz continuous and generate a 
posterior-mean measurable cost of information acquisition that rationalizes 
the dataset. NIAS is immediate by \cite{caplin2015revealed}, Theorem 1. For 
NIPMC, let $A_i\in\mathcal{A}$, and define $F_0^i$, $F_i$, $G_i$, and 
$\beta_i$ as in Axiom \ref{NIPMC}. Let $(F_i^*,D_{A_i}^*)$ solve 
(\ref{DMprob}), and generate $\sigma_{A_i}$. For each $i$ and each 
$a_i\in A_i$, define $A_i^*$ as a minimal set of acts needed to 
generate the support of $F_i$, by defining $a_i^*: A_i\rightarrow A_i$
such that 
$$a_i^*(a_i)=a_i^*(a_i^\prime)\iff z_{\sigma_{A_i}}(a_i)=z_{\sigma_{A_i}}(a_i^\prime)$$
and letting $A_i^*$ be the set of distinct values that $a_i^*$ takes. 
Next, for each $a_i\in A_i^*$, define
$$\tilde{D}_{A_i}^*(a_i\vert z)=\sum_{a_i^\prime \in A_i:a_i^*(a_i^\prime)=a_i} D_{A_i}^*(a_i^\prime\vert z)$$
$$\tilde{\sigma}_{A_i}(a_i)=\sum_{a_i^\prime \in A_i:a_i^*(a_i^\prime)=a_i}\sigma_{A_i}(a_i^\prime)$$
We then define the CDF $\hat{F}_i$ by 
\[
\hat{F}_i(z)=\sum_{a_i\in A_i^*} f_i(z_{\sigma_{A_i}}(a_i))\int_0^z \frac{\tilde{D}_{A_i}^*(a_i|s)}{\tilde{\sigma}_{A_i}(a_i)}dF_i^*(s)
\]
adopting the convention that $0/0=0$.

Suppose that $I_{F_0,F_i^*}(z^*)=0$. If $z_{\sigma_{A_i}}(a_i)< z^*$, 
then $\sigma_{A_i}(a_i\vert z)>0$ only if $z\leq z^*$, and therefore 
$\tilde{D}^*_{A_i}(a_i\vert z)>0$ only if $z\leq z^*$ as well. Conversely,
if $z_{\sigma_{A_i}}(a_i)> z^*$, then $\sigma_{A_i}(a_i\vert z)=0$ for 
all $z< z^*$. As a result,
\[
\int_0^{z^*} \frac{\tilde{D}^*_{A_i}(a_i\vert s)}{\tilde{\sigma}_{A_i}(a_i)}dF^*_i(s)=\mathbf{1}[z_{\sigma_{A_i}}(a_i)\leq z^*]
\]
while at the same time, for $a_i$ such that $z_{\sigma_{A_i}}(a_i)\leq z^*$,
by (\ref{Bayes}) and the Definitions \ref{revealedmean} and \ref{revealeddecision},
\begin{align*}
    \int_0^{z^*} \int_0^z \frac{\tilde{D}^*_{A_i}(a_i\vert s)}{\tilde{\sigma}_{A_i}(a_i)}dF^*_i(s)dz & =\int_0^{z^*}\int_0^z \sum_{a_i^\prime \in A_i:a_i^*(a_i^\prime)=a_i}\frac{\sigma_{A_i}(a_i^\prime\vert s)}{\tilde{\sigma}_{A_i}(a_i)}dF_0(s)dz\\
    & =\int_0^{z^*} \sum_{a_i^\prime \in A_i:a_i^*(a_i^\prime)=a_i}\frac{\sigma_{A_i}(a_i^\prime)}{\tilde{\sigma}_{A_i}(a_i)}\mathbf{1}[z_{\sigma_{A_i}}(a_i^\prime)\leq z]dz=\int_0^{z^*} \mathbf{1}[z_{\sigma_{A_i}}(a_i)\leq z]dz
\end{align*}
Thus for all 
$z^*\in \{z: I_{F_0,F_i^*}(z)=0\}$, 
\begin{align*}
    \int_0^{z^*} \hat{F}_i(z)dz &=\int_0^{z^*} \sum_{a_i\in A_i^*} f_i(z_{\sigma_{A_i}}(a_i))\int_0^z \frac{\tilde{D}^*_{A_i}(a_i\vert s)}{\tilde{\sigma}_{A_i}(a_i)}dF^*_i(s)dz\\ 
    &= \int_0^{z^*} \sum_{a_i\in A_i^*}f_i(z_{\sigma_{A_i}}(a_i))\mathbf{1}[z_{\sigma_{A_i}}(a_i)\leq z]dz\\
    &= \int_0^{z^*}F_i(z)dz\\
    &= \int_0^{z^*}F^i_0(z)dz
\end{align*}
where the second equality is by the above argument, and the last 
equality is from the fact that $I_{F_0^i,F_i}(z^*)=0$. Moreover, by 
construction, $\mbox{supp}(\hat{F}_i)\subseteq \mbox{supp}(F_i^*)$. 
Therefore, by Lemma \ref{DM2019}, $\hat{F}_i$ solves (\ref{DMprob}) for prior 
$F_0^i$. 

Similarly, define the CDF $\hat{G}_i$ by
\begin{equation}
    \label{hatg}
    \hat{G}_i(z)\coloneqq \sum_{j=1}^n \sum_{a_j\in A_j^*} \frac{g_i(z_{\sigma_{A_j}}(a_j))}{N(z_{\sigma_{A_j}}(a_j))}\int_0^z \frac{\tilde{D}^*_{A_j}(a_j\vert s)}{\tilde{\sigma}_{A_j}(a_j)}dF^*_j(s) 
\end{equation}
where $N(z)$ is the number of menus for which 
$z\in \mbox{supp}(F_{\sigma_{A_j}})$, again adopting the convention that 
$0/0=0$. 

By the same argument as for $\hat{F}_i$, for all $z^*\in \{z: I_{F_0,F_i^*}(z)=0\}$,
\begin{align*}
    \int_0^{z^*} \hat{G}_i(z)dz &=\int_0^{z^*} \sum_{j=1}^n \sum_{a_j\in A_j^*} \frac{g_i(z_{\sigma_{A_j}}(a_j))}{N(z_{\sigma_{A_j}}(a_j))}\int_0^z \frac{\tilde{D}^*_{A_j}(a_j\vert s)}{\tilde{\sigma}_{A_j}(a_j)}dF^*_j(s)dz\\
    &= \int_0^{z^*} \sum_{j=1}^n \sum_{a_j\in A_j^*} \frac{g_i(z_{\sigma_{A_j}}(a_j))}{N(z_{\sigma_{A_j}}(a_j))}\mathbf{1}[z_{\sigma_{A_j}}(a_j)\leq z]dz\\
    &= \int_0^{z^*} G_i(z)dz\\
    &\leq \int_0^{z^*} F_0^i(z)dz
\end{align*}
where the last inequality holds with equality at $z^*=1$. By Lemma 
\ref{DM2019}, $\hat{F}_i$ is optimal with respect to (\ref{DMprob}), 
and so
\begin{equation}
    \label{newopt}
\int_0^1 [\phi_{A_i}(z)+c(z)]d\hat{F}_i(z)\geq \int_0^1 [\phi_{A_i}(z)+c(z)]d\hat{G}_i(z),\forall i
\end{equation}

Multiplying (\ref{newopt}) by $\beta_i$ and summing over $i$, we get
\begin{equation}
    \label{weightednewopt}
    \sum_{i=1}^n \beta_i \int_0^1 \phi_{A_i}(z)d\hat{F}_i(z) +\int_0^1 c(z)d(\sum_{i=1}^n \beta_i \hat{F}_i(z))\geq \sum_{i=1}^n \beta_i \int_0^1 \phi_{A_i}(z)d\hat{G}_i(z) +\int_0^1 c(z)d(\sum_{i=1}^n \beta_i \hat{G}_i(z))
\end{equation}
Notice that $\{\hat{G}_i\}_{i=1}^n$ is a reallocation of posterior means of $\{\hat{F}_i\}_{i=1}^n$:
\begin{align*}
    \sum_{i=1}^n \beta_i \hat{G}_i(z) &=\sum_{i=1}^n \beta_i \sum_{j=1}^n \sum_{a_j\in A_j^*} \frac{g_i(z_{\sigma_{A_j}}(a_j))}{N(z_{\sigma_{A_j}}(a_j))}\int_0^z \frac{\tilde{D}^*_{A_j}(a_j\vert s)}{\tilde{\sigma}_{A_j}(a_j)}dF^*_j(s) \\
    &=\sum_{j=1}^n \sum_{a_j\in A_j^*} \frac{\sum_{i=1}^n \beta_i g_i(z_{\sigma_{A_j}}(a_j))}{N(z_{\sigma_{A_j}}(a_j))}\int_0^z \frac{\tilde{D}^*_{A_j}(a_j\vert s)}{\tilde{\sigma}_{A_j}(a_j)}dF^*_j(s) \\
    &=\sum_{j=1}^n \sum_{a_j\in A_j^*} \frac{\sum_{i=1}^n \beta_i f_i(z_{\sigma_{A_j}}(a_j))}{N(z_{\sigma_{A_j}}(a_j))}\int_0^z \frac{\tilde{D}^*_{A_j}(a_j\vert s)}{\tilde{\sigma}_{A_j}(a_j)}dF^*_j(s) \\
    &=\sum_{i=1}^n \beta_i \sum_{j=1}^n \sum_{a_j\in A_j^*} \frac{f_i(z_{\sigma_{A_j}}(a_j))}{N(z_{\sigma_{A_j}}(a_j))}\int_0^z \frac{\tilde{D}^*_{A_j}(a_j\vert s)}{\tilde{\sigma}_{A_j}(a_j)}dF^*_j(s) \\
    &=\sum_{i=1}^n \beta_i \sum_{a_i\in A_i^*} f_i(z_{\sigma_{A_i}}(a_i))\int_0^z \frac{\tilde{D}^*_{A_i}(a_i\vert s)}{\tilde{\sigma}_{A_i}(a_i)}dF^*_i(s) \\
    &= \sum_{i=1}^n \beta_i \hat{F}_i(z)
\end{align*}
where the penultimate equality is from a counting argument. Therefore,
\[
\int_0^1 c(z)d(\sum_{i=1}^n \beta_i \hat{F}_i(z))=\int_0^1 c(z)d(\sum_{i=1}^n \beta_i \hat{G}_i(z))
\]
and so from (\ref{weightednewopt}), 
\[
\sum_{i=1}^n \beta_i \int_0^1 \phi_{A_i}(z)d\hat{F}_i(z)\geq \sum_{i=1}^n \beta_i \int_0^1 \phi_{A_i}(z)d\hat{G}_i(z)
\]
Recalling that $A_i^*$ samples a single representative act
for each $z\in\mbox{supp}(F_{\sigma_{A_i}})$,
\begin{align*}
    \int_0^1 \phi_{A_i}(z)d\hat{F}_i(z) &= \sum_{a_i\in A_i^*} f_i(z_{\sigma_{A_i}}(a_i))\int_0^z \phi_{A_i}(z)\frac{\tilde{D}_{A_i}(a_i\vert z)}{\tilde{\sigma}_{A_i}(a_i)}dF^*_i(z)\\
    &= \sum_{a_i\in A_i^*} f_i(z_{\sigma_{A_i}}(a_i))\int_0^z u(a_i,z)\frac{\tilde{D}_{A_i}(a_i\vert z)}{\tilde{\sigma}_{A_i}(a_i)}dF^*_i(z)\\
    &= \sum_{a_i\in A_i^*} f_i(z_{\sigma_{A_i}}(a_i)) u(a_i,z_{\sigma_{A_i}}(a_i))\\
    &= \int_0^1 \phi_{A_i}(z)dF_i(z)
\end{align*}
Moreover, $\hat{G}_i\succeq G_i$ since there may be multiple values of 
$s$ in the definition (\ref{hatg}) at which $D_{A_j}^*(a_j\vert s)>0$, 
which are then garbled into $z_{\sigma_{A_j}}(a_j)$, over the latter values 
of which $G_i$ has support. Thus,
\[
\sum_{i=1}^n \beta_i \int_0^1 \phi_{A_i}(z)dF_i(z)=\sum_{i=1}^n \beta_i \int_0^1 \phi_{A_i}(z)d\hat{F}_i(z)\geq \sum_{i=1}^n \beta_i \int_0^1 \phi_{A_i}(z)d\hat{G}_i(z)\geq \sum_{i=1}^n \beta_i \int_0^1 \phi_{A_i}(z)dG_i(z)
\]
and so Axiom \ref{NIPMC} is satisfied.

Conversely, suppose that $\{\sigma_A\}_{A\in \mathcal{A}}$ is a dataset that 
satisfies NIAS and NIPMC. By Lemma \ref{finiteineq}, it is sufficient to 
determine that (\ref{allineq}) holds by checking that it holds at $z\in Z$.
Let $Z^*_A\coloneqq\{z: I_{F_0,F_{\sigma_{A}}}(z)=0\}\cap Z$. By Lemma 
\ref{exanteNIPMC}, there do not exist nonnegative $\beta_{A,B}(a,b)$,
where $A,B\in\mathcal{A}$, $a\in A$, and $b\in B$, such that for all 
$A\in \mathcal{A}$ and $z^*\in Z^*_A$,
\begin{multline}
\label{prod1}
    \sum_{B\in \mathcal{A}} \int_0^{z^*}\sum_{s\in Z} [\sum_{b\in B}\sum_{\{a\in A: z_{\sigma_{A}}(a)\leq z\}}\beta_{(A,B)}(a,b)\sigma_{A}(a\vert s)\\
    -\sum_{a\in A}\sum_{\{b\in B: z_{\sigma_{B}}(b)\leq z\}}\beta_{(B,A)}(b,a)\sigma_{B}(b\vert s)]f_0(s)dz\geq 0
\end{multline}
with equality at $z^*=1$, 
\begin{equation}
\label{prod2}
\sum_{j=1}^n \sum_{z \in Z}\sum_{a_i \in A_i} \sum_{a_j \in A_j } \beta_{(i,j)}(a_i, a_j) \sigma_{A_i}(a_i|z)f_0(z)=\sum_{j=1}^n \sum_{z \in Z}\sum_{a_i \in A_i} \sum_{a_j \in A_j } \beta_{(j,i)}(a_j, a_j) \sigma_{A_j}(a_j|z)f_0(z)
\end{equation}
and
\begin{equation}
    \label{prod3}
    \sum_{A\in \mathcal{A}}\sum_{B\in\mathcal{A}}\sum_{z\in Z}\sum_{b\in B}\sum_{a\in A}\beta_{(A,B)}(a,b) [u(a,z)-u(b,z)]\sigma_{A}(a\vert z)f_0(z)<0
\end{equation}

Inequality (\ref{prod1}) can be rewritten as 
\begin{multline}
\label{prod4}
    \sum_{B\in\mathcal{A}} [\sum_{b\in B}\sum_{a\in A: z_{\sigma_{A}}(a)\leq z^*} \beta_{(A,B)}(a,b)\sigma_A(a)(z^*-z_{\sigma_A}(a))\\
     -\sum_{a\in A}\sum_{b\in B: z_{\sigma_{B}}(b)\leq z^*} \beta_{(B,A)}(b,a)\sigma_B(b)(z^*-z_{\sigma_B}(b))]\geq 0
\end{multline}
while (\ref{prod2}) can be rewritten as 
\begin{equation}
    \label{prod5}
    \sum_{j=1}^n \sum_{a_i \in A_i} \sum_{a_j \in A_j } \beta_{(i,j)}(a_i, a_j) \sigma_{A_i}(a_i)=\sum_{j=1}^n \sum_{a_i \in A_i} \sum_{a_j \in A_j } \beta_{(j,i)}(a_j, a_j) \sigma_{A_j}(a_j)
\end{equation}
and (\ref{prod3}) can be rewritten as 
\begin{equation}
    \label{prod6}
    \sum_{A\in \mathcal{A}}\sum_{B\in\mathcal{A}}\sum_{b\in B}\sum_{a\in A}\beta_{(A,B)}(a,b)[u(a,z_{\sigma_A}(a))-u(b,z_{\sigma_A}(a))]\sigma_{A}(a)<0
\end{equation}

Let the matrix $\mathbf{A}$ as being indexed by rows $i=(A,B,a,b)$, and columns
indexed by $j=(z^*,A^\prime)$ as
\begin{equation}
    \mathbf{A}_{i,j}=\begin{cases}
    (z^*-z_{\sigma_A}(a))\sigma_A(a), & z^*>0, z_{\sigma_A}(a)\leq z^*, A=A^\prime\neq B\\
    -(z^*-z_{\sigma_A}(a))\sigma_A(a), & z^*>0, z_{\sigma_A}(a)\leq z^*, A\neq A^\prime=B\\
    \sigma_A(a), & z^*=0, A=A^\prime\neq B\\
    -\sigma_A(a), &  z^*=0, A\neq A^\prime=B\\
    0, & \mbox{otherwise}
    \end{cases}
\end{equation}
Meanwhile, define the vector $\mathbf{b}$ with entries indexed by $i=(A,B,a,b)$,
such that 
\[
\mathbf{b}_i=[u(a,z_{\sigma_A}(a))-u(b,z_{\sigma_A}(a))]\sigma_A(a).
\]
Then (\ref{prod4})-(\ref{prod6}) is equivalent to there not existing a
vector $\beta$, where $\beta_i=\beta_{(A,B)}(a,b)$ for $i=(A,B,a,b)$, such that
\begin{enumerate}[(i)]
    \item $\beta\geq 0$,
    \item $(\mathbf{A}^T\beta)_j\geq 0$, with equality for $j$ where 
    $z^*\in \{0,1\}$
    \item $\mathbf{b}\cdot \beta<0$
\end{enumerate}
By Lemma \ref{farkas}, there exists vector $\lambda$ with entries $j$ 
corresponding to $(z,A)$ such that $\mathbf{A}\lambda \leq \mathbf{b}$
and $\lambda_j\geq0$ for $j$ such that $z^*\notin\{0,1\}$ where, letting 
\[
\Lambda_A(z)\coloneqq \lambda_{(0,A)}+\sum_{z^*\in Z^*_A} \lambda_{(z^*,A)} (z^*-z)\mathbf{1}[z^*\geq z]
\]
we have that for all $A,B\in\mathcal{A}$, $a\in A$, and $b\in B$,
\begin{equation}
\label{optimuminequality}
    \sigma_A(a)[\Lambda_A(z_{\sigma_A}(a))-\Lambda_B(z_{\sigma_A}(a))]\leq \sigma_A(a) [u(a,z_{\sigma_A}(a))-u(b,z_{\sigma_A}(a))]
\end{equation}
Define $c:[0,1]\rightarrow \mathbb{R}$ by 
\begin{equation}
    \label{IDcost}
    c(z)\coloneqq- \max_{A\in\mathcal{A}, a\in A} \{u(a,z)- \Lambda_A(z)\}
\end{equation}
Notice that $c$ need not be concave, and hence not be canonical, since 
$\Lambda_A(z)$ is convex. However, $c$ is Lipschitz continuous, since 
$\lambda_{z^*,A}$ is finite.

To show that $c$ rationalizes the dataset, by 
Lemma \ref{DMT12} and the argument in the proof of Lemma \ref{DM2019}, 
$\sigma_A$ is optimal if there exists $P^*_A:[0,1]\rightarrow\mathbb{R}$ 
convex such that
\begin{enumerate}
    \item $P^*_A(z)\geq \max_{a\in A} u(a,z)+c(z),\forall z\in[0,1]$,
    \item $P^*_A(z_{\sigma_A}(a))=u(a,z_{\sigma_A}(a))+c(z_{\sigma_A}(a)), \forall a \mbox{ such that } \sigma_A(a)>0$, and 
    \item For any interval $[z_1,z_2]$ on which 
    $I_{F_0,F_{\sigma_A}}(z)>0,\forall z\in(z_1,z_2)$, $P^*_A$ is affine with
    slope $p^*_A(z)$.
\end{enumerate}
Letting 
\[
P^*_A(z)=\Lambda_A(z)
\]
we get that $c(z_{\sigma_A}(a))=P^*_A(z_{\sigma_A}(a))-u(a,z_{\sigma_A}(a))$
for all $a\in\mbox{supp}(\sigma_A)$ by (\ref{optimuminequality}). Moreover,
$P^*_A$ is affine for $z\in[z_1,z_2]$ with slope 
$-\sum_{z^*\in Z^*_A: z^*\geq z_2}\lambda_{(z^*,A)}$. Since 
$\lambda_j\geq 0$ for all $j$ such that $z^*\neq 1$, $p^*_A$ is increasing, and 
so $P^*_A$ is convex. Lastly, for all $z\in [0,1]$, $\hat{a}\in A$,
\begin{align*}
u(\hat{a},z)+c(z) & \leq u(\hat{a},z)-\max_{a\in A} u(a,z) +\Lambda_A(z)\\
&\leq \Lambda_A(z)\\
&=P^*_A(z)
\end{align*}
Thus $P_A^*$ satisfies the conditions of Theorem 1 of \cite{dworczak2019simple}.

\end{proof}

\section{Proofs from Section 4}

\subsection{Proof of Proposition \ref{cs}}

    Suppose for contradiction, for some $z_1, z_2\in supp (F_{\sigma_{A_2}})$, where $\mathcal{I}_{F_0,F_{\sigma_{A_1}}}(z) >0$ for $z_1< z< z_2,$, there exists $\hat{z} \in  supp (F_{\sigma_{A_1}})\cap (z_1, z_2).$ 

There exists $\alpha \in (0,1)$ such that 
$\hat{z}=\alpha z_1 + (1-\alpha)z_2 $. Consider the function $H$ defined 
as follows:%
\footnote{The construction of such $H$ for a continuum of states is more involved. See, for instance, \cite{mensch2024monopoly} for examples of how to construct improvements of distributions with a continuum of states.}
    \begin{equation*}
        H(z)=
        \left\{
        \begin{array}{ll}
        0, & z< z_1\\
        \alpha, &  z_1 \le z < \hat{z}\\
        \alpha-1, & \hat{z} \le z \le  z_2\\
        0,& z\ge z_2
        \end{array}
        \right.
\end{equation*}

 Now consider a reallocation of $F_{\sigma_{A_1}}$ and 
$F_{\sigma_{A_2}}$ to $G_1\coloneqq F_{\sigma_{A_1}}+ \beta H$ and 
$G_2 \coloneqq F_{\sigma_{A_2}}- \beta H$ with equal weights, where 
$\beta$ is small enough to ensure that that $G_1$ and $G_2$ are 
well-defined distribution functions.

Let $\psi\coloneqq \phi_{A_1}- \phi_{A_2}.$ Note that
\begin{align*}
   \int \psi d( F_{\sigma_{A_1}}) < \int \psi d( F_{\sigma_{A_i}} - \beta H)= \int \psi d( G_1)\\  
\end{align*}
due to global convexity of $\psi$ and strict convexity at $\hat{z}.$

This implies
\begin{align*}
      \int \phi_{A_1} d(F_{\sigma_{A_1}})+ \int \phi_{A_2} d(F_{\sigma_{A_2}})&=  \int \phi_{A_1} d(F_{\sigma_{A_1}})+ \int \phi_{A_2} d(G_2+G_1- F_{\sigma_{A_1}})\\
    &= \int \phi_{A_2} d(G_2) + \int [\phi_{A_1}- \phi_{A_2}]d(F_{\sigma_{A_1}})\\
   & ~+ \int \phi_{A_2} d(G_1) - \int \phi_{A_1} d(G_1)+ \int \phi_{A_1} d(G_1)\\
    &=\int \phi_{A_2} d(G_2) +\int \phi_{A_1} d(G_1)+ \int \psi d(F_{\sigma_{A_1}})- \int \psi d( G_1)\\
    &<\int \phi_{A_2} d(G_2) +\int \phi_{A_1} d(G_1)
\end{align*}
which violates NIPMC.

\subsection{Proof of Theorem \ref{continuumexist}}

For the direction of necessity, Axiom \ref{NIAS} is satisfied
for the same reason as in the finite case. For Axiom \ref{NIPMC}, similar 
to the finite case, we want to define $A_i^*$ as a minimal set of acts 
needed to generate the support of $F_i$; however, we need to make sure that 
we construct something that is well-defined and measurable. First, we note 
that, due to the assumption that there are no two acts in $A_i$ with the 
same payoffs, nor are any dominated, they can be ordered by the slope of 
$u_z(\cdot,z)$, in increasing order. It is therefore without loss of generality to assume $a_i=u_z(\cdot,z)$. Thus defined, as $X$ is compact,
the set of such slopes is bounded.

Define $F_0^i,F_i,G_i,\beta_i,F_i^*$, and $D_{A_i}^*$ as in the proof of 
Theorem \ref{costexist}, the latter being well-defined almost everywhere
by the Lebesgue differentiation theorem. Notice that for each $z$, the 
measure $\int_0^z D_{A_i}^*(\cdot\vert s)dF_i^*(s)$ is absolutely 
continuous with respect to $\sigma(\cdot)$, since 
$\int_0^1 D_{A_i}^*(\cdot\vert s)dF_i^*(s)=\sigma(\cdot)$. Moreover,
by Axiom \ref{NIAS}, each act $a_i$ such that $z_{\sigma_{A_i}}(a_i)=z$
is optimal given $z$. 

Now note that the set of acts that share the same value of 
$z_{\sigma_{A_i}}$ is \emph{convex}: that is, for all 
$a_i^1,a_i^2,a_i^3\in A_i\cap\mbox{supp}(\sigma_{A_i})$, where 
$a_i^2\in[a_i^1,a_i^3]$, 
$$z_{\sigma_{A_i}}(a_i^1)=z_{\sigma_{A_i}}(a_i^3)=z^*\implies z_{\sigma_{A_i}}(a_i^2)=z^*$$
This follows immediately from the convexity of $\phi_{A_i}$: the subgradient
of $\phi_{A_i}$, which is identical to the set of values that 
$u_z(\cdot,z)$ can take if $a$ is optimal at $z$ is a convex set,
and that there is a unique value of $a_i\in A_i$ with a given slope. 

One can therefore define as before $a_i^*$ as providing a selection
from the values of $a_i$ that yield the same $z_{\sigma_{A_i}}(a_i)$.
Then, for each $a_i\in A_i^*$, define
$$\tilde{D}_{A_i}^*(a_i\vert z)=\int_{a_i^\prime \in A_i:a_i^*(a_i^\prime)=a_i} D_{A_i}^*(a_i^\prime\vert z)da_i^\prime$$
$$\tilde{\sigma}_{A_i}(a_i)=\int_{a_i^\prime \in A_i:a_i^*(a_i^\prime)=a_i}\sigma_{A_i}(a_i^\prime)da_i^\prime$$
Notice that, so constructed, both $\tilde{D}_{A_i}([0,a_i]\vert z)$ and 
$\tilde{\sigma}_{A_i}([0,a_i])$ are nonnegative, increasing, right-continuous, and take value $1$ at $a_i=1$; therefore, they are valid probability measures. 

We then define
\[
\hat{F}_i(z)\coloneqq \int_{a_i\in A_i^*} \int_0^z \frac{dD^*_{A_i}(a_i\vert s)}{d\sigma_{A_i}(a_i)}dF^*_i(s)dF_i(z_{\sigma_{A_i}}(a_i))
\]
\[
\hat{G}_i(z)\coloneqq \sum_{j=1}^n \int_{a_j\in A_j^*} \int_0^z \frac{dD^*_{A_j}(a_j\vert s)}{d\sigma_{A_j}(a_j)}dF^*_j(s)dG_i(z_{\sigma_{A_j}}(a_j))
\]
where the Radon-Nikodym derivative 
$\int_0^z \frac{dD^*_{A_j}(a_j\vert s)}{d\sigma_{A_j}(a_j)}dF^*_j(s)$ is well-%
defined by absolute continuity, and the integration over $A_i^*$ is well-%
defined due to the measurability of the SDSC dataset. The proof then 
proceeds identically to that of the discrete case.

For the direction of sufficiency, we present the following analogue of
Lemma \ref{exanteNIPMC}.

\begin{lemma}
 Let $A_1,...,A_n$ be a sequence of menus in $\mathcal{A}$, 
    $\beta_{(i,j)}:A_i\times A_j\rightarrow \mathbb{R}_+^2$ continuous,
    and measures $\tau_i$ over $A_i$ for $i,j=1,...n$. If Axioms \ref{NIAS} 
    and \ref{NIPMC} hold, and:
    \begin{enumerate}
        \item For all $i$ and 
    $z^*\in \{z: I_{F_0,F_{\sigma_{A_i}}}(z)=0\}$,
    \begin{multline}
    \label{reallocMPCnobindc}
    \sum_{j=1}^n \int_0^{z^*}\int_0^1 [\int_{A_j}\int_{\{a_i: z_{\sigma_{A_i}}(a_i)\leq z\}}\beta_{(i,j)}(a_i,a_j)d\sigma_{A_i}(a_i\vert s)d\tau_j(a_j)\\
    -\int_{A_i}\int_{\{a_j: z_{\sigma_{A_j}}(a_j)\leq z\}}\beta_{(j,i)}(a_j,a_i)d\sigma_{A_j}(a_j\vert s)d\tau_i(a_i)]dF_0(s)dz\geq 0
    \end{multline}
    with equality at $z=1$.
    
    \item For all $i$, 
       \begin{multline}
        \label{reallocsamepriorc}    \sum_{j=1}^n \int_0^1\int_{A_i} \int_{A_j} \beta_{(i,j)}(a_i, a_j) d\sigma_{A_i}(a_i|z)d\tau_j(a_j) dF_0(z)=\\
        \sum_{j=1}^n \int_0^1\int_{A_i} \int_{A_j} \beta_{(j,i)}(a_j, a_j) d\sigma_{A_j}(a_j|z)d\tau_i(a_i) dF_0(z)
    \end{multline}
    \end{enumerate}
    Then 
    \begin{equation}
        \label{exanteNIPMCeqcontinuum}
        \sum_{i=1}^n\sum_{j=1}^n\int_0^1\int_{A_j}\int_{A_i}\beta_{(i,j)}(a_i,a_j) [u(a_i,z)-u(a_j,z)]d\sigma_{A_i}(a_i\vert z)d\tau_j(a_j) dF_0(z)\geq0
    \end{equation}
\end{lemma}

The proof is identical to that of Lemma \ref{exanteNIPMC}, and so is omitted. 
Notice that the linear functional of measures $\Phi(\sigma_{A_i},\sigma_{A_j})$ 
as given by the left-hand side of (\ref{reallocMPCnobindc}) and in 
(\ref{reallocsamepriorc}) are continuous in the weak* topology. By the Convex 
Cone Alternative (\cite{aliprantis2006infinite}, Corollary 5.84), using the 
same construction as in Lemma \ref{farkas} and in the proof of Theorem 
\ref{costexist}, for each $A\in\mathcal{A}$, there exist measures
$\lambda_{A}$, $\lambda_{(0,A)}\in\mathbb{R}$, and 
$\bar{\lambda}_{A}\in \mathbb{R}$ such that almost everywhere with respect
to $\sigma_A$, for all $A,B\in\mathcal{A}$, $a\in A$, $b\in B$,
\begin{equation}
\Lambda_A(z_{\sigma_A}(a))-\Lambda_B(z_{\sigma_A}(a))\leq u(a,z_{\sigma_A}(a))-u(b,z_{\sigma_A}(a))
\end{equation}
where
\[
\Lambda_A(z)\coloneqq \bar{\lambda}_A(1-z_{\sigma_A}(a))+\lambda_{(0,A)}+\int_{[z_{\sigma_A}(a),1]\cap Z^*_A} (z-z_{\sigma_A}(a))d\lambda_A(z)
\]
Defining
\[
c(z)\coloneqq -\max_{A\in \mathcal{A},a\in A}\{u(a,z)-\Lambda_A(z)\}
\]
\[
P^*_A(z)\coloneqq \Lambda_A(z)
\]
yields the desired result by the same steps as in the proof of Theorem 
\ref{costexist}.

\section{Linear Programming Algorithm for the Concavity of $c$}

As shown in the sufficiency direction of the proof of Theorem 
\ref{costexist}, we construct the functions $c$ and $P^*_A$ by use of 
the coefficients $\lambda_{(z^*,A)}$. It is quite straightforward to 
see that between any two consecutive values of $z^*\in Z^*_A$, the
function $c$, as defined by \ref{IDcost}, is concave. The potential
issue is at $z^*\in Z^*_A$: if $\lambda_{(z^*,A)}>0$, there may be 
a kink in $c$ at $z^*$ due to the change in the slope at that point,
if it is menu $A$ that is generating $c$ at $z^*$.

One way that this problem can be avoided is if, in addition to the 
set of inequalities given by (\ref{optimuminequality}), the coefficients
$\{\lambda_{(z^*,A)}\}$ are such that, whenever $A$ is the menu that 
generates $c$ at $z^*$, then $\lambda_{(z^*,A)}=0$. Thus we can check
whether the inequalities in (\ref{optimuminequality}), as well as
the inequalities, for $z\in Z$
\begin{equation}
    \label{cavcheck1}
    \lambda_{(z,A_{i_z})}=0
\end{equation}
\begin{equation} 
\label{cavcheck2}
\Lambda_{A_{i_z}}(z)-\Lambda_B(z)\le \phi_{A_{i_z}}(z)- \phi_{B}(z),\forall B\in\mathcal{A}
\end{equation}
for some values of $\{i_z\}_{z\in Z}$, where, by convention,
$\lambda_{(z,A)}=0$ if $z\notin Z^*_A$. This involves checking at most 
$n\vert Z\vert$ inequalities, for each possible arrangement of values of 
$i_z$: for each $z\in Z$, for a given $i_z$, there are $n-1$ inequalities of the 
form (\ref{cavcheck2}), as well as one equality of the form (\ref{cavcheck1}).
There are $n^{\vert Z\vert}$ possible such arrangements. Thus,
this algorithm for checking these sufficient conditions for the concavity
of $c$ will have the computational complexity of running $n^{\vert Z\vert}$
linear programs, each with at most $n\vert Z\vert$ inequalities.

\end{document}